%% file: example_paper.tex
\documentclass{article}

\usepackage{amsmath}

\input{chapters/pkg}

\conftitlerunning{Budgeted Active Experimentation for Treatment Effect Estimation from Observational and Randomized Data}

\begin{document}

\twocolumn[
  \conftitle{Budgeted Active Experimentation for Treatment Effect Estimation from Observational and Randomized Data}

\input{chapters/authors} 
  

  \confkeywords{Causal Inference, Active Learning, Treatment Effect Estimation}

  \vskip 0.3in
]

\printAffiliationsAndNotice{\confEqualContribution}

\begin{abstract}

Estimating heterogeneous treatment effects is central to data-driven decision-making, yet industrial applications often face a fundamental tension between limited randomized controlled trial (RCT) budgets and abundant but biased observational data collected under historical targeting policies. Although observational logs offer the advantage of scale, they inherently suffer from severe policy-induced imbalance and overlap violations, rendering standalone estimation unreliable. We propose a \textit{budgeted active experimentation} framework that iteratively enhances model training for causal effect estimation via active sampling. By leveraging observational priors, we develop an acquisition function targeting uplift estimation uncertainty, overlap deficits, and domain discrepancy to select the most informative units for randomized experiments. We establish finite-sample deviation bounds, asymptotic normality via martingale Central Limit Theorems (CLTs), and minimax lower bounds to prove information-theoretic optimality. Extensive experiments on industrial datasets demonstrate that our approach significantly outperforms standard randomized baselines in cost-constrained settings.
\end{abstract}

\input{chapters/intro}
\input{chapters/prelim}

\input{chapters/math}

\input{chapters/theory}
\input{chapters/exp}

\input{chapters/conclu}

\clearpage

\bibliography{example_paper}
\bibliographystyle{conf2026}
\input{chapters/appendix}

\end{document}

%% file: chapters/pkg.tex
\usepackage{microtype}
\usepackage{graphicx}
\usepackage{subcaption}
\usepackage{booktabs} 

\usepackage{hyperref}

\usepackage[preprint]{conf2026}

\usepackage{amsmath}
\usepackage{amssymb}
\usepackage{mathtools}
\usepackage{amsthm}

\usepackage[capitalize,noabbrev]{cleveref}

\theoremstyle{plain}
\newtheorem{theorem}{Theorem}[section]
\newtheorem{proposition}[theorem]{Proposition}
\newtheorem{lemma}[theorem]{Lemma}
\newtheorem{corollary}[theorem]{Corollary}
\theoremstyle{definition}

\newtheorem{assumption}[theorem]{Assumption}
\theoremstyle{remark}

\usepackage[textsize=tiny]{todonotes}

\usepackage[utf8]{inputenc}
\usepackage{booktabs}  
\usepackage{tabularx}             
\usepackage{pifont}   
\newcommand{\cmark}{\ding{51}}
\newcommand{\xmark}{\ding{55}}

\usepackage{makecell}
\usepackage{multirow}

%% file: chapters/authors.tex
\newcommand{\confsetaffiliationnumber}[2]{%
  \ifcsname the@affil#1\endcsname\else
    \newcounter{@affil#1}%
  \fi
  \setcounter{@affil#1}{#2}%
  \ifnum\value{@affiliationcounter}<#2\relax
    \setcounter{@affiliationcounter}{#2}%
  \fi
}
 \confsetsymbol{equal}{*}          
\confsetaffiliationnumber{fi}{1} 
\confsetaffiliationnumber{se}{2} 

  \begin{confauthorlist}
    \confauthor{Jiacan Gao}{equal,se}
    \confauthor{Xinyan Su}{equal,fi}
    \confauthor{Mingyuan Ma}{comp}
    \confauthor{Yiyan Huang}{dwq}
    \confauthor{Xiao Xu}{fi}
    \confauthor{Xinrui Wan}{fi}
    \confauthor{Tianqi Gu}{fi}
    \confauthor{Enyun Yu}{fi}
    \confauthor{Jiecheng Guo}{fi}
    \confauthor{Zhiheng Zhang}{shangcai1,shangcai2}
  \end{confauthorlist}

\confaffiliation{fi}{Didi Chuxing, Beijing, China}
\confaffiliation{se}{School of Statistics, East China Normal University, Shanghai, China}

  \confaffiliation{comp}{School of Mathematics and Statistics, Beijing Jiaotong University, Beijing, China}
  \confaffiliation{dwq}{School of Computing and Information Technology, Great Bay University‌, Guangdong, China}
  \confaffiliation{shangcai1}{School of Statistics and Data Science, Shanghai University
of Finance and Economics, Shanghai 200433, P.R. China}
  \confaffiliation{shangcai2}{Institute of Data Science and Statistics, Shanghai University of Finance and Economics,
Shanghai 200433, P.R. China}

  \confcorrespondingauthor{Zhiheng Zhang}{zhangzhiheng@mail.shufe.edu.cn}

%% file: chapters/intro.tex
\section{Introduction}

Estimating heterogeneous treatment effects (HTE), such as the conditional average treatment effect (CATE),
is a central problem in causal inference and data-driven decision making.
Accurate HTE estimation underpins personalized marketing, recommendation systems,
clinical decision support, and public policy design,
where decisions must be tailored to individual characteristics rather than population averages.
In large-scale real-world systems, however, HTE estimation is constrained by a fundamental asymmetry in data sources:
randomized controlled trials (RCTs) provide unbiased and reliable causal signals,
but are expensive, slow to deploy, and severely limited in sample size;
observational data (OBS), in contrast, are abundant and high-dimensional,
but are generated under historical targeting policies that induce selection bias,
policy-driven imbalance, and violations of overlap \citep{hatt2022combining, colnet2024causal}.

This asymmetry has motivated a growing literature on learning treatment effects from multiple data sources.
Classical work focuses on \emph{causal identification} from observational data under strong assumptions
such as ignorability and positivity,
while more recent approaches seek to \emph{combine} OBS and RCT data through reweighting or doubly robust estimators
\citep{cheng2021adaptive}.
In industrial settings, however, these approaches face a critical limitation:
historical policies are often near-deterministic,
creating regions of the covariate space where observational data provide essentially no counterfactual information.
In such regimes, directly using OBS outcomes for identification is potentially unreliable,
while running large-scale uniform RCTs is prohibitively costly.

This tension naturally shifts the focus from estimation to \emph{experiment design}.
Rather than asking how to extract causal effects from biased observational logs,
a more operational question is:
\emph{given abundant observational data and a strict budget for randomized experiments,
where should one run RCTs to most efficiently learn heterogeneous treatment effects?}
This question lies at the intersection of causal inference and active learning,
but differs fundamentally from classical active learning:
here, querying a unit does not merely reveal a label,
but requires actively assigning a treatment and observing a causal outcome.

Existing active learning approaches have explored this problem from two largely separate perspectives.
One line of work studies active or budgeted sampling \emph{within experimental settings}~\citep{zhangactive,kato2024active,ghadiri2023finite},
where treatments can be freely assigned but the RCT sample size is limited.
Another line focuses on active learning \emph{within observational studies}~\citep{wen2025enhancingtreatmenteffectestimation,gao2025causalepigpredictionorientedactivelearning},
where treatment assignments are fixed and the budget corresponds to labeling or outcome acquisition.
In contrast, the practically relevant regime where \emph{observational data and experimental design coexist}—
abundant biased logs alongside a small, adaptively collected RCT—
has received comparatively little theoretical and methodological attention.
In particular, it remains unclear how observational data should guide adaptive experiment design
while retaining finite-sample validity and statistical inference guarantees.

In this paper, we propose a \emph{budgeted active experimentation framework}
that addresses this gap by cleanly separating the roles of observational and randomized data.
The core idea is that \emph{observational data inform experiment design,
while randomized experiments provide causal estimation}.
Specifically, observational logs are used to learn a shared representation,
to diagnose overlap deficits induced by historical policies,
and to identify covariate regions under-represented in the current experimental sample.
These signals are combined into a multi-criteria acquisition function
that actively selects which units to enroll into RCTs.
The final estimator, however, is identified purely by randomized experiments,
ensuring robustness to arbitrary confounding in the observational source.

This design induces a nontrivial statistical regime:
the experimental data are adaptively collected, non-i.i.d., and depend on past outcomes.
We show that, despite adaptivity,
randomization induces a martingale structure that protects unbiasedness,
while active sampling shapes the information matrix and governs statistical efficiency.
Perhaps counterintuitively, we prove that even the most aggressive active experimentation strategy
cannot beat a $\sqrt{d/B}$ rate in general,
where $d$ is the effective dimension of heterogeneity and $B$ is the RCT budget.
Active learning does not create information out of thin air,
but it can substantially improve constants by repairing overlap and stabilizing the design. Our main contributions are summarized as follows:

\textbf{(i)} We propose a principled framework for budgeted active experimentation
that integrates observational and randomized data by separating experiment design from causal identification.

\textbf{(ii)} We provide a comprehensive theoretical analysis for adaptively collected RCT data,
including finite-sample deviation bounds, asymptotic normality via martingale CLTs,
and minimax lower bounds establishing near-optimality.

\textbf{(iii)} We clarify of the precise role of observational data in causal learning:
OBS improves \emph{where} to randomize, not \emph{what} is identified,
yielding a robust and interpretable path to cost-efficient HTE estimation.

The remainder of the paper is organized as follows.
Section~\ref{sec:related} reviews related work on active learning and treatment effect estimation.
Section~\ref{sec:problem} formalizes the two-source causal learning problem.
Section~\ref{sec:method} presents the proposed active experimentation algorithm.
Section~\ref{sec:theory} develops the theoretical guarantees under adaptive, non-i.i.d.\ experimentation.
Section~\ref{sec:experiment} reports empirical results on large-scale real-world data.
Section~\ref{sec:conclusion} concludes with implications and future directions.

%% file: chapters/prelim.tex
\begin{figure}[t]
    \centering
    \includegraphics[width=0.98\linewidth]{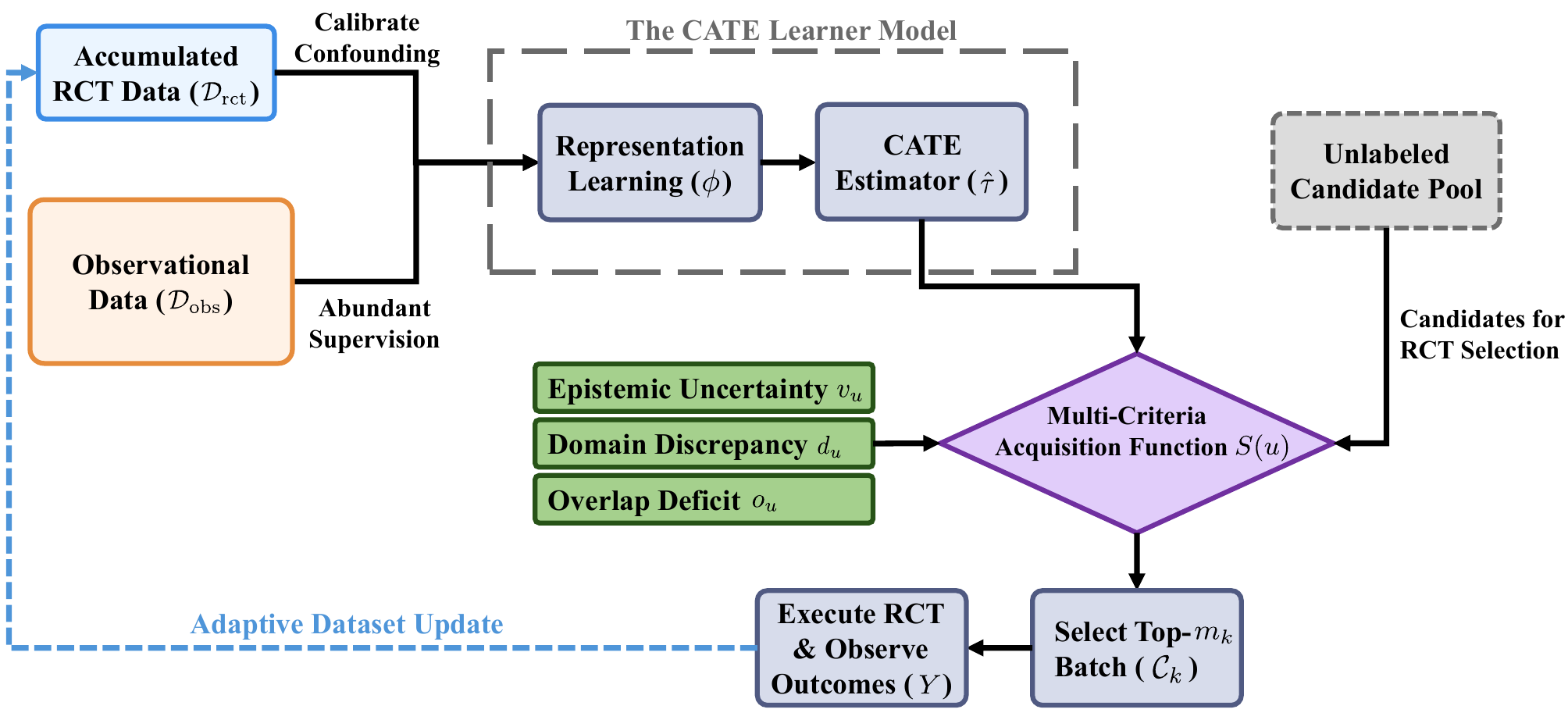}
    \caption{\textbf{Overview of the proposed Budgeted Active Experimentation framework for OBS-RCT fusion.}  We first train a CATE learner using abundant observational logs, then iteratively select a small batch of units from an unlabeled candidate pool $\mathcal{D}_{\mathrm{pool}}$ for randomized experiments under a fixed budget. At each iteration, a multi-criteria acquisition function $S(u)$ scores and ranks candidates by balancing three signals: \textbf{epistemic uncertainty} ($v_u$), \textbf{domain discrepancy} ($d_u$), and \textbf{overlap deficit} ($o_u$). We select the top-ranked units to run RCTs, add the new outcomes to the labeled set, and update the learner to reduce CATE estimation error.
}
    \label{fig:framework}
\end{figure}

\section{Problem Formulation}\label{sec:problem}
\textbf{Notation.}
Uppercase letters denote random variables and lowercase denote their realizations.
Let $X\in\mathcal{X}\subseteq\mathbb{R}^d$ be covariates, $T\in\{0,1\}$ a binary treatment,
and $Y\in\{0,1\}$ a binary outcome (the formulation extends to any bounded $Y\in[0,1]$).
For any distribution $\mathcal{Q}$, we write $\mathbb{E}_{\mathcal{Q}}[\cdot]$ and $\mathbb{P}_{\mathcal{Q}}(\cdot)$
for expectation and probability under $\mathcal{Q}$.
We adopt the Neyman--Rubin potential outcomes framework~\citep{rubin2005causal}.
Each unit has potential outcomes $Y(1)$ and $Y(0)$ under treatment and control, respectively.
The observed outcome follows $Y = Y(T)$.

\begin{assumption}{(SUTVA).}\label{assump:sutva}
(i) \emph{No interference}: a unit's potential outcomes are unaffected by other units' assignments.
(ii) \emph{Consistency}: $Y = T\,Y(1) + (1-T)\,Y(0)$.
\end{assumption}

For the \emph{target estimand and evaluation metric}, our estimand is the \textbf{conditional average treatment effect (CATE)} \citep{shalit2017estimating, zhong2022descn, gao2025causalepigpredictionorientedactivelearning}:
\begin{equation}
	\label{eq:2}
	\tau(x) \;=\; \mathbb{E}\!\left[\, Y(1) - Y(0) \,\middle|\, X=x \right], \quad x\in\mathcal{X}.
\end{equation}
Define $\mu_t(x)\triangleq \mathbb{E}[Y(t)\mid X=x]$ for $t\in\{0,1\}$, so that $\tau(x)=\mu_1(x)-\mu_0(x)$.
We evaluate generalization over a \emph{target population} whose covariate marginal distribution is $\mathbb P_X$.
Given an estimator $\hat{\tau}$, we use the PEHE-style risk \citep{Hill01012011, wen2025enhancingtreatmenteffectestimation, gao2025causalepigpredictionorientedactivelearning}:
\begin{equation}
	\label{eq:targetrisk}
	\mathcal{R}(\hat{\tau}) \triangleq 
	\sqrt{\mathbb{E}_{X\sim \mathbb{P}_X}\!\left[(\hat{\tau}(X)-\tau(X))^2\right]}.
\end{equation}

We have access to two-source data: \ \
(i) \textbf{Observational log (OBS).}
	An observational dataset
	$\mathcal{D}_{\mathrm{obs}}=\{(X_i^{\mathrm{obs}},T_i^{\mathrm{obs}},Y_i^{\mathrm{obs}})\}_{i=1}^{n_{\mathrm{obs}}}
	\overset{\mathrm{i.i.d.}}{\sim}\mathbb{P}_{\mathrm{obs}}$,
	collected under a historical (possibly non-random) assignment policy.
	Let the induced observational propensity be
	$e_{\mathrm{obs}}(x)\triangleq \mathbb{P}_{\mathrm{obs}}(T^{\mathrm{obs}}=1\mid X^{\mathrm{obs}}=x),$
	which may be highly imbalanced and even (nearly) deterministic in some regions. We write $\mathbb P_{\mathrm{obs}}$ as shorthand for $\mathbb P_{\mathbb P_{\mathrm{obs}}}$.\\
	(ii) \textbf{Unlabeled candidate pool (RCT-POOL).}
	An unlabeled pool $\mathcal{D}_{\mathrm{pool}}=\{X_j^{\mathrm{pool}}\}_{j=1}^{n_{\mathrm{pool}}}
	\overset{\mathrm{i.i.d.}}{\sim}\mathbb P_X$,
	from which we may \emph{query} a small subset of units to run randomized experiments.
	Units in $\mathcal{D}_{\mathrm{pool}}$ receive no treatment unless queried. 

Because OBS treatments follow a historical policy, global positivity in $\mathcal{D}_{\mathrm{obs}}$
(i.e., $0<e_{\mathrm{obs}}(x)<1$ for all $x$) may fail.
We therefore \emph{do not} assume global overlap for OBS; instead, we will use randomized experiments on queried units
to obtain counterfactual information in weak(non)-overlap regions.

Moreover, we view $\mathcal D_{\mathrm{pool}}=\{X^{\mathrm{pool}}_j\}_{j=1}^{n_{\mathrm{pool}}}$ as an unlabeled sample from the \emph{target} covariate marginal $\mathbb P_X$, whereas $\mathcal D_{\mathrm{obs}}=\{(X^{\mathrm{obs}}_i,T^{\mathrm{obs}}_i,Y^{\mathrm{obs}}_i)\}_{i=1}^{n_{\mathrm{obs}}}$ is drawn from a different joint law $\mathbb P_{\mathrm{obs}}$ induced by a historical (possibly near-deterministic) targeting policy. Accordingly, the covariate marginals need not match: in general $\mathbb P_{X^{\mathrm{obs}}}\neq \mathbb P_X$ and $\mathrm{supp}(\mathbb P_{X^{\mathrm{obs}}})\subseteq \mathrm{supp}(\mathbb P_X)$, reflecting selection bias and potential overlap violations in the observational log. Our design therefore treats $\mathcal D_{\mathrm{pool}}$ as the population whose risk is evaluated, while using $\mathcal D_{\mathrm{obs}}$ only to guide where to randomize (and not for causal identification in weak-overlap regions).

\textbf{Active Experimentation Protocol (Adaptive RCT).}
We sequentially run experiments for $K$ rounds (or until the budget is exhausted).
Let $\mathcal{D}_{\mathrm{rct}}^{(0)}=\emptyset$ initially.
At round $k$, an adaptive strategy selects an index set $\mathcal{C}_k\subseteq[n_{\mathrm{pool}}]$ from the remaining pool.
For each selected unit with covariate $X$, we assign treatment
$T^{\mathrm{rct}} \sim \mathrm{Bern}\!\big(p_k(X)\big), 
	\qquad p_k:\mathcal{X}\to[f_{\min},f_{\max}],$
where the randomization probability $p_k(\cdot)$ is \emph{known} and satisfies $0<f_{\min}\le f_{\max}<1$.
We then observe $Y^{\mathrm{rct}}$ and record the quadruple $(X,T^{\mathrm{rct}},Y^{\mathrm{rct}},p_k(X))$.
The experimental dataset is updated as
\begin{equation}\label{update_procedure}
	\mathcal{D}_{\mathrm{rct}}^{(k)} \leftarrow \mathcal{D}_{\mathrm{rct}}^{(k-1)} \cup 
	\{(X_j^{\mathrm{pool}},T_j^{\mathrm{rct}},Y_j^{\mathrm{rct}},p_k(X_j^{\mathrm{pool}})):\, j\in \mathcal{C}_k\}.
\end{equation}
Since $\mathcal{C}_k$ and $p_k(\cdot)$ may depend on the history, $\mathcal{D}_{\mathrm{rct}}^{(k)}$ is adaptively collected and thus non-i.i.d.

\begin{assumption}{(Randomization and positivity in RCT).}\label{assump:rct-rand}
Conditioned on the history up to round $k$ and covariates $X$, the RCT assignment is randomized:
$T^{\mathrm{rct}} \perp (Y(0),Y(1)) \mid (X,\text{history})$ and
$\mathbb{P}(T^{\mathrm{rct}}=1\mid X,\text{history}) = p_k(X)\in[f_{\min},f_{\max}]$.
\end{assumption}

\begin{assumption}{(Outcome invariance across sources).}\label{assump:invariance}
The conditional potential outcome distributions are shared across sources:
for each $t\in\{0,1\}$ and $x\in\mathcal{X}$,
\[
\mathbb P_{\mathrm{obs}}\big(Y(t)\mid X=x\big)={\mathbb P_X}\big(Y(t)\mid X=x\big).
\]
Equivalently, the target CATE $\tau(x)$ is common to both sources.
\end{assumption}
\begin{assumption}{(Ignorability for OBS; used when leveraging OBS causally).}\label{assump:obs-ign}
When we use OBS for causal identification (e.g., via propensity-based corrections),
we assume \emph{conditional ignorability}:
$(Y(0),Y(1)) \perp T^{\mathrm{obs}} \mid X^{\mathrm{obs}}$ under $\mathbb{P}_{\mathrm{obs}}$.
\end{assumption}

\textbf{Objective.}
Let $\pi$ denote an \emph{adaptive experiment design} that generates the sequence $\{(\mathcal{C}_k,p_k)\}_{k\ge1}$
based on the observed history.
Given a total experimental budget $B$ (the number of queried units), we aim to design $\pi$ and a corresponding estimator
$\hat{\tau}_{\pi}$ trained using both $\mathcal{D}_{\mathrm{obs}}$ and $\mathcal{D}_{\mathrm{rct}}^{(K)}$ such that the target risk~\eqref{eq:targetrisk} is minimized:
\begin{equation}
	\label{obj}	
	\min_{\pi}\; \mathcal{R}\big(\hat{\tau}_{\pi}\big)
	\quad\text{s.t.}\quad \sum_{k=1}^{K} |\mathcal{C}_k| \le B.
\end{equation}
Equivalently, one may consider the sample complexity $N_{\mathrm{RCT}}(\varepsilon)$:
the minimum budget required to achieve $\mathcal{R}(\hat{\tau}_{\pi})\le \varepsilon$ with high probability.

Assumptions~\ref{assump:sutva}--\ref{assump:obs-ign} are standard in causal learning but play distinct roles here:
(i) SUTVA ensures a well-defined unit-level causal model;
(ii) Assumption~\ref{assump:rct-rand} guarantees unbiased identification from queried RCT samples and prevents extreme-variance estimators via $[f_{\min},f_{\max}]$;
(iii) Assumption~\ref{assump:invariance} enables pooling information across sources (diagnosable by comparing covariate distributions and outcome model residuals across domains);
(iv) Assumption~\ref{assump:obs-ign} is only needed if OBS is used for causal correction---if it is questionable, OBS can be used more conservatively (e.g., for representation learning / warm-starting),
while identification and calibration rely primarily on the randomized samples.

%% file: chapters/math.tex
\section{Methodology}\label{sec:method}

\begin{algorithm}[tb]
   \caption{Active Sampling for OBS-RCT Fusion}
   \label{alg:active_fusion}
\begin{algorithmic}[1]
    \REQUIRE Observational Data $\mathcal{D}_{\mathrm{obs}}$, Unlabeled Pool $\mathcal{D}_{\mathrm{pool}}$, Max Query Batch Size $M$, Budget $B$, Max Rounds $K$.
    \ENSURE Final CATE Estimator $\hat{\tau}_{\pi}$.
    \STATE {\bfseries Initialize:} $\mathcal{D}_{\mathrm{rct}} \leftarrow \emptyset$, $k \leftarrow 1$.
    \STATE {\bfseries Representation Learning:} Train feature encoder $\phi: \mathcal{X} \to \mathcal{H}$ on $\mathcal{D}_{\mathrm{obs}} \cup \mathcal{D}_{\mathrm{pool}}$.
   \WHILE{$|\mathcal{D}_{\mathrm{rct}}| < B$ and $k \leq K$}
       \STATE $m_k \leftarrow \min(M, B - |\mathcal{D}_{\mathrm{rct}}|)$ 
       
       \STATE \textbf{1. Update Scoring Components:}
        \STATE \quad Train domain classifier $g_\xi(\phi(x))$ to discriminate $\mathcal{D}_{\mathrm{pool}}$ (label 1) from $\mathcal{D}_{\mathrm{obs}} \cup \mathcal{D}_{\mathrm{rct}}$ (label 0).
        \STATE \quad Update CATE predictions $\mathcal{M}=\{\hat{\tau}_j(x)\}_{j=1}^E$. \hfill \COMMENT{$E$: \# of stochastic passes}
        
        \STATE \quad Estimate propensity $\hat{e}_{\mathrm{obs}}(\phi(x)) \approx \mathbb{P}_{\mathrm{obs}}(T^{\mathrm{obs}}=1 \mid \phi(x))$
        
        \STATE \textbf{2. Acquisition Scoring:}
       \FOR{candidate $u \in \mathcal{D}_{\mathrm{pool}}$}
           \STATE $v_u \leftarrow \mathrm{Var}\left( \{ \hat{\tau}_j(\phi(u)) \}_{j=1}^E \right)$ \hfill 
           \STATE $d_u \leftarrow \sigma (g_\xi(\phi(u)))$ \hfill 
           \COMMENT{$\sigma(\cdot)$: Sigmoid Function}
           \STATE $o_u \leftarrow 2 \cdot|\hat{e}_{\mathrm{obs}}(\phi(u))-0.5|$ \hfill 
           \STATE Compute Score: \hfill \COMMENT{$\eta(\cdot)$: Rank Function}
           \STATE $S(u) \leftarrow \alpha \cdot \eta(v_u) + \beta \cdot \eta(d_u) + \gamma \cdot \eta(o_u)$ \hfill
       \ENDFOR

       \STATE \textbf{3. Selection \& Experimentation:}
       \STATE Select candidates $\mathcal{C}_k \leftarrow \text{Top-}m_k(S(\cdot), \mathcal{D}_{\mathrm{pool}})$.
       \FOR{$X_i \in \mathcal{C}_k$}
           \STATE Perform RCT: Assign $T_i^{\mathrm{rct}} \sim \mathrm{Bern}(p_k(X_i))$, observe $Y_i^{\mathrm{rct}}$.
            \STATE $\mathcal{D}_{\mathrm{rct}} \leftarrow \mathcal{D}_{\mathrm{rct}} \cup \{(X_i, T_i^{\mathrm{rct}}, Y_i^{\mathrm{rct}}, p_k(X_i))\}$.
       \ENDFOR

       \STATE \textbf{4. Update:} $\mathcal{D}_{\mathrm{pool}} \leftarrow \mathcal{D}_{\mathrm{pool}} \setminus \mathcal{C}_k$; \quad $k \leftarrow k + 1$.
   \ENDWHILE

   \STATE \textbf{Return} $\hat{\tau}_{\pi}$ trained on $\mathcal{D}_{\mathrm{obs}} \cup \mathcal{D}_{\mathrm{rct}}$.
\end{algorithmic}
\end{algorithm}

We study a two-source setting where (i) a large observational log $\mathcal{D}_{\mathrm{obs}}$ provides abundant but potentially biased and weak-overlap supervision, and
(ii) an unlabeled pool $\mathcal{D}_{\mathrm{pool}}$ represents the target population from which we may enroll a small number of units into randomized experiments.
Our objective is to leverage a limited experimental budget $B$ to acquire an adaptively designed RCT dataset $\mathcal{D}_{\mathrm{rct}}$, and then learn a final CATE estimator
$\hat{\tau}_\pi$ from the augmented data $\mathcal{D}_{\mathrm{obs}}\cup\mathcal{D}_{\mathrm{rct}}$.

The key difficulty is \emph{where} to run RCTs. Running experiments uniformly at random wastes budget on regions where OBS already provides reliable information,
while failing to repair regions where OBS is systematically uninformative (e.g., extreme propensities) or where the learned model extrapolates.
We therefore cast experiment design as a \emph{pool-based active learning} problem for causal effect estimation:
each queried unit provides \emph{both} a randomized treatment assignment and an outcome label, which can directly reduce CATE error.

Algorithm~\ref{alg:active_fusion} summarizes our procedure. Below we explain the intuition behind each module, focusing on three canonical
sources of CATE estimation error that are particularly severe in OBS--RCT fusion:\\
    \textbf{(i) Epistemic uncertainty} of the CATE model due to finite data and high-dimensional covariates;\\
    \textbf{(ii) Domain discrepancy} between the biased OBS log and the actively collected RCT sample (and thus the target pool);\\
    \textbf{(iii) Overlap deficit} in OBS, where near-deterministic historical targeting makes counterfactual information essentially absent.\\

Each error source is addressed by an explicit design choice: covariate shift is handled by querying from the unlabeled pool $\mathcal{D}_{\mathrm{pool}}$ to restore coverage of the target marginal $\mathbb{P}_X$; weak or missing overlap is mitigated by targeted randomization rather than reweighting observational outcomes; and adaptivity bias is controlled by design-based randomization guarantees. Here, \emph{uncertainty} refers not only to estimation variance conditional on a fixed design, but also to design-induced uncertainty arising from which regions of the covariate space receive randomized evidence. 
Our active strategy reduces this uncertainty by allocating experimental budget to covariate regions with the largest marginal contribution to CATE risk, yielding tighter error bounds for a fixed RCT cost. As a result, uncertainty is actively shaped—and provably reduced—by the experimental design, rather than passively inherited from OBS–RCT fusion. Our acquisition score explicitly combines proxies for these three error sources, so that each RCT query yields maximal marginal value toward reducing the target PEHE risk.


To mitigate instability in high-dimensional industrial settings, we employ a shared Multilayer Perceptron (MLP) encoder, defined as $\phi: \mathcal{X} \to \mathcal{H}$, to map raw covariates into a dense latent representation. This shared backbone serves as a common foundation for (i) propensity estimation, (ii) CATE modeling, and (iii) domain discrimination, so that all components operate in a unified lower-dimensional space.

\textbf{$v_u$: Backbone CATE learner and ensemble for uncertainty.}
Since single causal models typically underestimate epistemic uncertainty~\citep{NEURIPS2022_675e371e}, we rely on ensemble disagreement to quantify the lack of knowledge. While such disagreement can be captured via Deep Ensembles~\citep{lakshminarayanan2017simple}, we employ Monte Carlo (MC) Dropout~\citep{gal2016dropout} for computational efficiency. We perform $E$ stochastic forward passes to obtain a set of CATE predictions $\{\hat{\tau}_j\}_{j=1}^E$. The uncertainty score is defined as the variance of these estimates:
\begin{equation}
    v_u \triangleq \text{Var}\left( \{ \hat{\tau}_j(\phi(u)) \}_{j=1}^E \right)
\end{equation}
A large $v_u$ implies significant disagreement among the stochastic sub-models, indicating that querying $u$ yields high information gain by reducing epistemic uncertainty.

\textbf{$d_u$: Domain discrimination for representativeness and distributional alignment.} To prevent the active set from drifting solely towards decision boundaries and to mitigate the selection bias inherent in $\mathcal{D}_{\text{obs}}$, we explicitly enforce distributional alignment. We employ a domain classifier $g_\xi$ operating on the representation space, trained to distinguish the target pool $\mathcal{D}_{\text{pool}}$ from the current training set $\mathcal{D}_{\text{current}}=\mathcal{D}_{\text{obs}} \cup \mathcal{D}_{\text{rct}} $. The score is defined as the predicted probability of belonging to the target pool (label 1):
\begin{equation} 
    d_u \triangleq \mathbb{P}(\text{domain}=1 \mid \phi(u); \xi) = \sigma(g_\xi(\phi(u)))
\end{equation}
where $g_\xi(\cdot)$ denotes the logit output and $\sigma(\cdot)$ is the sigmoid function. Theoretically, $g_\xi$ acts as a density-ratio estimator where the logit approximates $\log(p_{\text{pool}}(\phi) / p_{\text{current}}(\phi))$. Consequently, prioritizing high $d_u$ targets samples in regions under-represented by the current source data, minimizing covariate shift and ensuring robust generalization.

\textbf{$o_u$: Propensity-based overlap deficit from OBS.}
To mitigate estimator instability caused by extreme observational propensities (positivity violations), we explicitly target regions lacking common support. We utilize a propensity model $\hat{e}_{\mathrm{obs}}(\phi(u))\;\approx\;\mathbb{P}_{\mathrm{obs}}(T^{\mathrm{obs}}=1\mid \phi(u))$, pre-trained on $\mathcal{D}_{\text{obs}}$, to define the deficit score:
$
    o_u \triangleq 2 \cdot |\hat{e}_{\mathrm{obs}}(\phi(u)) - 0.5|.
$
Samples with high $o_u$ correspond to regions where the historical policy was deterministic. Prioritizing these instances for RCT labeling effectively ``repairs'' the overlap deficit by injecting counterfactual information exactly where the observational signal is weakest. 


\textbf{Rank-Normalized Multi-Criteria Acquisition Function. } The three signals $(v_u, d_u, o_u)$ can have very different scales and may change over rounds.
Instead of brittle scale-dependent normalization, we use a simple \emph{rank-based} map $\eta(\cdot)$,
computed over the current pool:
$
    \eta(a_u)
    \triangleq
    \frac{1}{|\mathcal{D}_{\mathrm{pool}}|}\sum_{u'\in\mathcal{D}_{\mathrm{pool}}}
    \mathbb{I}\!\left(\psi(a_{u'}) \le \psi(a_u)\right)
    \in[0,1],
$
where $\psi(\cdot)$ is a monotone transform used for ranking.
Since the constructed signals $v_u$, $d_u$, and $o_u$ are all non-negative and positively correlated with the informativeness of a sample, we simply use identity transform $\psi(a)=a$ for ranking.
For each pool candidate $u$, we compute acquisition score:
\begin{equation}
    \label{eq:acq_score}
    S(u)
    \;\triangleq\;
    \alpha\cdot \eta(v_u)
    \;+\;
    \beta\cdot \eta(d_u)
    \;+\;
    \gamma\cdot \eta(o_u),
\end{equation}
and select the top-$m_k$ candidates. $\alpha, \beta$, and $\gamma$ balance the trade-off between uncertainty , discrepancy, and counterfactual coverage. To ensure scale consistency across these heterogeneous metrics, we apply rank normalization prior to weighted aggregation.

\subsection{Adaptive RCT execution and dataset update}\label{sec:method:rct}
After selecting $\mathcal{C}_k$, we conduct randomized experiments for each $X_i\in \mathcal{C}_k$.
We allow a covariate-dependent randomization policy $f_k(\cdot)$ (e.g., to incorporate operational constraints),
and enforce positivity by clipping: $p_i \;\triangleq\; \mathrm{Clip}(f_k(X_i), f_{\min}, f_{\max}),$
where $\mathrm{Clip}(z,a,b)\triangleq \min\{\max\{z,a\},b\}$. We then draw $T_i\sim\mathrm{Bern}(p_i)$, observe $Y_i$, and update as in~\eqref{update_procedure}. Storing $p_i$ is essential for downstream learning/inference when randomization probabilities are not constant.
Finally, we remove $\mathcal{C}_k$ from $\mathcal{D}_{\mathrm{pool}}$ and repeat until the budget is exhausted. We reduce the selection of optimal $p$ in Proposition~\ref {prop:optimal_p}.

%% file: chapters/theory.tex
\section{Theoretical Analysis}\label{sec:theory}

This section provides theoretical guarantees for the \emph{estimation component} of Algorithm~\ref{alg:active_fusion}
under the adaptively collected RCT data.
Our goal is to make three claims precise: (i) \textbf{Finite-sample validity}: an (near-) unbiased effect estimator and a non-asymptotic error bound that remains valid under \emph{adaptive, non-i.i.d.} RCT sampling. (ii) \textbf{Asymptotic inference}: a central limit theorem (CLT) showing asymptotic normality enabling confidence intervals. (iii) \textbf{Fundamental limits}: a minimax lower bound showing that the achieved rate is information-theoretically optimal (up to logs). A key theme is that {randomization protects unbiasedness, while active selection shapes the information matrix and hence the error rate.}
We first formalize an analyzable estimator that matches the RCT protocol in Algorithm~\ref{alg:active_fusion} and
is standard in modern causal ML: regress an \emph{orthogonalized pseudo-outcome} onto a representation.

\textbf{A tractable estimator for adaptively collected RCT samples} Let $B$ denote the total RCT budget and index queried units in chronological order $t=1,\dots,B$.
Write the $t$-th queried sample as $(X_t,T_t,Y_t,p_t)$ where $p_t\in[f_{\min},f_{\max}]$ is the (known) assignment probability
used when unit $t$ was experimented. Let $\{\mathcal{F}_{t}\}_{t\ge 0}$ be the filtration where $\mathcal{F}_{t}$ contains
$\mathcal{D}_{\mathrm{obs}}$, the full covariate pool $\mathcal{D}_{\mathrm{pool}}$, and all RCT data up to time $t$.
The adaptive querying policy (via $\mathcal{C}_k$, $f_k$) implies that $X_t$ and $p_t$ may depend on $\mathcal{F}_{t-1}$. We define the RCT pseudo-outcome
$\widetilde{Y}_t
    \;\triangleq\;
    \frac{T_t\,Y_t}{p_t}
    \;-\;
    \frac{(1-T_t)\,Y_t}{1-p_t}.$
Intuitively, $\widetilde{Y}_t$ converts each single randomized trial into a noisy but unbiased ``label'' of the treatment effect at $X_t$.

Let $\phi(x)\in\mathbb{R}^d$ be a fixed feature map (e.g., the learned embedding $\phi$ in Algorithm~\ref{alg:active_fusion}).
For theory, we analyze a realizable linear CATE model in this representation:
$\tau(x) = \langle \theta_\star, \phi(x)\rangle,~\theta_\star\in\mathbb{R}^d.$
This assumption is a standard and useful ``microscope'': it isolates the statistical difficulty caused by
\emph{adaptive designs + small budgets}, while keeping the analysis transparent.
(When $\tau(\cdot)$ is not exactly linear, the bounds below become bounds on the best linear approximation plus an approximation error term.) Given $\{(X_t,T_t,Y_t,p_t)\}_{t=1}^B$, we estimate $\theta_\star$ by (optionally regularized) least squares on pseudo-outcomes:
\begin{equation}
    \label{eq:ridge}
    \widehat{\theta}_\lambda
    \;\triangleq\;
    \arg\min_{\theta\in\mathbb{R}^d}\;
    \sum_{t=1}^{B}\Big(\widetilde{Y}_t-\langle \theta,\phi(X_t)\rangle\Big)^2
    \;+\;\lambda\|\theta\|_2^2,
\end{equation}
and output $\widehat{\tau}_\lambda(x)\triangleq \langle \widehat{\theta}_\lambda,\phi(x)\rangle$.

\begin{algorithm}[t]
\caption{Orthogonalized Linear CATE Estimation on Adaptive RCT Data}
\label{alg:theory_est}
\begin{algorithmic}[1]
\STATE \textbf{Input:} Adaptive RCT data $\{(X_t,T_t,Y_t,p_t)\}_{t=1}^B$, feature map $\phi(\cdot)$, ridge $\lambda\ge 0$.
\STATE \textbf{Output:} CATE estimator $\widehat{\tau}_\lambda(\cdot)$.
\FOR{$t=1,\dots,B$}
    \STATE Compute pseudo-outcome $\widetilde{Y}_t \leftarrow \frac{T_t Y_t}{p_t}-\frac{(1-T_t)Y_t}{1-p_t}$.
\ENDFOR
\STATE Form $V_\lambda \leftarrow \lambda I_d + \sum_{t=1}^{B}\phi(X_t)\phi(X_t)^\top$ and $b\leftarrow \sum_{t=1}^{B}\phi(X_t)\widetilde{Y}_t$.
\STATE Solve $\widehat{\theta}_\lambda \leftarrow V_\lambda^{-1}b$.
\STATE Return $\widehat{\tau}_\lambda(x)\leftarrow \langle \widehat{\theta}_\lambda,\phi(x)\rangle$.
\end{algorithmic}
\end{algorithm}

Algorithm~\ref{alg:active_fusion} specifies \emph{how the $X_t$'s are chosen} (active sampling) and
\emph{how $p_t$ is assigned} (bounded randomization).
Algorithm~\ref{alg:theory_est} specifies a transparent estimator for analyzing the statistical consequences of that adaptive design.
In practice, one may replace the linear regressor by a neural CATE learner; the key objects in the proofs below are
(i) the unbiased pseudo-outcome, and (ii) an information matrix that quantifies how well the queried points cover the representation space.

\paragraph{Role of observational data in the estimator.}
It is important to clarify that the causal identification of $\tau(x)$ is carried exclusively by randomized experiments.
Observational data are \emph{not} directly used in the estimating equations.
Instead, $\mathcal D_{\mathrm{obs}}$ enters the procedure in two structural ways:
(i) it defines the representation $\phi(\cdot)$ in which the CATE is approximated,
and (ii) it shapes the adaptive experiment design that determines the queried covariates $\{X_t\}_{t=1}^B$.
All theoretical guarantees are therefore \emph{design-adaptive but identification-robust}:
they remain valid even when the observational assignment mechanism is arbitrarily confounded.

\subsection{Unbiasedness and finite-sample error bounds}
\label{sec:theory:finite}

We now state the first main result: despite \emph{adaptive} (non-i.i.d.) selection of $X_t$,
the RCT pseudo-outcome remains unbiased, which implies unbiasedness of the linear estimator (for $\lambda=0$)
and a sharp finite-sample deviation bound (for any $\lambda\ge 0$).

\begin{assumption}{(Predictable adaptive design).}
\label{assump:predictable}
For each $t$, the queried covariate $X_t$ and assignment probability $p_t$ are $\mathcal{F}_{t-1}$-measurable.
\end{assumption}

\begin{assumption}{(RCT randomization and boundedness).}
\label{assump:rct_bound}
Conditioned on $(X_t,p_t,\mathcal{F}_{t-1})$, treatment is randomized as
$T_t \sim \mathrm{Bern}(p_t)$ and $T_t \perp (Y_t(0),Y_t(1)) \mid (X_t,p_t,\mathcal{F}_{t-1})$.
Outcomes are bounded: $Y_t(0),Y_t(1)\in[0,1]$ almost surely, and $p_t\in[f_{\min},f_{\max}]$ with $0<f_{\min}\le f_{\max}<1$.
\end{assumption}

\begin{assumption}{(Linear realizability in representation space).}
\label{assump:linear}
There exists $\theta_\star\in\mathbb{R}^d$ such that $\tau(x)=\langle \theta_\star,\phi(x)\rangle$ for all $x$.
Moreover, $\|\theta_\star\|_2\le S$ and $\|\phi(x)\|_2\le L$ for all $x$.
\end{assumption}

\begin{lemma}[Unbiased pseudo-outcome under adaptive sampling]
\label{lem:unbiased_pseudo}
Under Assumptions~\ref{assump:predictable}--\ref{assump:rct_bound},
\begin{equation}
    \label{eq:unbiased_pseudo}
    \mathbb{E}\!\left[\widetilde{Y}_t \mid X_t,p_t,\mathcal{F}_{t-1}\right] \;=\; \tau(X_t),
    \qquad t=1,\dots,B.
\end{equation}
$   |\widetilde{Y}_t|
    \;\le\;
    \max\!\left\{\frac{1}{f_{\min}},\,\frac{1}{1-f_{\max}}\right\}
    \;\triangleq\; L_p.$
\end{lemma}

Equation~\eqref{eq:unbiased_pseudo} says: \emph{even if we cherry-pick who enters the experiment},
as long as we randomize treatment for the selected unit, each queried unit yields an unbiased, single-sample estimate of its own CATE.
A concrete example is coupon delivery:
even if we actively choose ``hard'' users (e.g., high-value users with extreme historical targeting),
a randomized coupon/no-coupon decision still produces an unbiased causal contrast at that user. This lemma is \emph{not} claiming that the observational estimator is unbiased, nor that any neural training procedure is unbiased.
It only uses RCT randomization and boundedness.
The adaptive selection affects \emph{where} we learn (which $X_t$'s), not \emph{whether} a queried unit is causally valid.

A potential {skeptic is that} ``But the learner chooses $X_t$ based on past outcomes; isn't this a form of selection bias that breaks unbiasedness?''
\emph{Response:} the selection changes the distribution of $X_t$, but unbiasedness in~\eqref{eq:unbiased_pseudo} is \emph{conditional} on $(X_t,p_t,\mathcal{F}_{t-1})$.
The key is that $T_t$ is randomized \emph{after} selection and is conditionally independent of potential outcomes.
Technically, $\widetilde{Y}_t-\tau(X_t)$ forms a martingale difference sequence. Under this preparation, we introduce the finite-sample bound as follows:

\begin{theorem}[Finite-sample unbiasedness and deviation bound]
\label{thm:finite}
Assume~\ref{assump:predictable}--\ref{assump:linear}.
Let $V_\lambda \triangleq \lambda I_d + \sum_{t=1}^{B}\phi(X_t)\phi(X_t)^\top$ and let $\widehat{\theta}_\lambda$ be defined in~\eqref{eq:ridge}.
Then: (i) \textbf{(Conditional unbiasedness for OLS).}
    If $\lambda=0$ and $V_0$ is invertible, then
    $
        \mathbb{E}\!\left[\widehat{\theta}_0 \mid \{X_t,p_t\}_{t=1}^B\right]=\theta_\star,
        \qquad
        \mathbb{E}\!\left[\widehat{\tau}_0(x)\mid \{X_t,p_t\}_{t=1}^B\right]=\tau(x),\ \forall x.
    $
     \textbf{(High-probability self-normalized bound).}
    Fix $\delta\in(0,1)$ and any $\lambda\ge 0$.
    There exists a constant $\sigma\le 2L_p$ (depending only on $f_{\min},f_{\max}$ and the boundedness of $Y$)
    such that with probability at least $1-\delta$,
    \begin{equation}
        \label{eq:self_norm}
        \big\|\widehat{\theta}_\lambda-\theta_\star\big\|_{V_\lambda}
        \;\le\;
        \underbrace{\sigma\sqrt{2\log\frac{\det(V_\lambda)^{1/2}}{\det(\lambda I_d)^{1/2}\,\delta}}}_{\text{stochastic term}}
    +\underbrace{\sqrt{\lambda}\,S}_{\text{regularization bias}}.
    \end{equation}
    Consequently, for every $x\in\mathcal{X}$,
    \begin{equation}
        \label{eq:pointwise}
        \big|\widehat{\tau}_\lambda(x)-\tau(x)\big|
        \;\le\;
        \beta_{B}(\delta)\cdot
        \sqrt{\phi(x)^\top V_\lambda^{-1}\phi(x)},
    \end{equation}
where $\beta_{B}(\delta)\triangleq
        \sigma\sqrt{2\log\frac{\det(V_\lambda)^{1/2}}{\det(\lambda I_d)^{1/2}\,\delta}}
        +\sqrt{\lambda}S.$
\end{theorem}

Theorem~\ref{thm:finite} quantifies a clean ``budget $\Rightarrow$ error'' tradeoff under \emph{adaptive} experimentation:
your estimation error is controlled by an \emph{information matrix} $V_\lambda$ built from the queried covariates.
The only price paid for adaptivity is that the $X_t$'s are not i.i.d.; the deviation bound still holds because the noise is a martingale difference. It is easy to confuse~\eqref{eq:pointwise} with classical i.i.d. regression bounds.
The crucial difference is: \emph{the $X_t$'s can be adversarially chosen by the learner itself, based on past outcomes.}
Theorem~\ref{thm:finite} remains valid in that setting; it does not require i.i.d. sampling of $X_t$.

Noteworthy, in very small $B$, $V_0$ can be ill-conditioned.
This is exactly why we stated~\eqref{eq:self_norm} for ridge $\lambda>0$ as well:
the bound decomposes into a \emph{stochastic term} plus an explicit \emph{regularization bias}.
This makes the tradeoff transparent and diagnosable in practice: if $V_0$ is poorly conditioned, increase $\lambda$ to stabilize,
and the theory tells you exactly how much bias you introduce.

\begin{corollary}[Finite-sample PEHE bound]
\label{cor:pehe}
Under the conditions of Theorem~\ref{thm:finite}, with probability at least $1-\delta$,
\begin{equation}
    \label{eq:pehe_bound}
    \mathcal{R}(\widehat{\tau}_\lambda)
    \;\le\;
    \beta_{B}(\delta)\cdot
    \sqrt{\mathbb{E}_{X\sim \mathbb{P}_X}\!\left[\phi(X)^\top V_\lambda^{-1}\phi(X)\right]}.
\end{equation}
In particular, if the queried design is \emph{well-conditioned} in the sense that
$V_0 \succeq \kappa B\, I_d$ for some $\kappa>0$, and $\Sigma_X\triangleq \mathbb{E}_{\mathbb{P}_X}[\phi(X)\phi(X)^\top]\preceq L^2 I_d$,
then (taking $\lambda=0$ for simplicity),
\begin{equation}
    \label{eq:rate_simple}
    \mathcal{R}(\widehat{\tau}_0)
    \le
    \beta_{B}(\delta)\cdot \sqrt{\frac{\mathrm{Tr}(\Sigma_X)}{\kappa B}}
    \lesssim
    \frac{L\,\sigma}{\sqrt{\kappa}}\sqrt{\frac{d\log(B/\delta)}{B}}.
\end{equation}
\end{corollary}

Corollary~\ref{cor:pehe} shows that the PEHE risk is governed by an \emph{integrated leverage} term
$\mathbb{E}_{\mathbb{P}_X}[\phi(X)^\top V_\lambda^{-1}\phi(X)]$.
Active sampling is useful precisely because it can shape $V_\lambda$:
selecting points that increase eigenvalues of $V_\lambda$ reduces this term.
This provides a direct theoretical rationale for using an ``uncertainty'' proxy (Algorithm~\ref{alg:active_fusion}, $v_u$):
in linear models, predictive variance is proportional to $\phi(u)^\top V_\lambda^{-1}\phi(u)$,
and ensemble disagreement is a practical surrogate for that quantity.

\subsection{Asymptotic normality under adaptive sampling}
\label{sec:theory:clt}

Finite-sample concentration ensures ``small error with high probability.''
For statistical inference (e.g., confidence intervals), we further need a CLT.
The non-i.i.d. nature of active sampling prevents a direct appeal to the classical i.i.d. CLT,
but a martingale CLT applies.

\begin{assumption}{(Stabilizing design and moments).}
\label{assump:stabilize}
As $B\to\infty$, the normalized information matrix converges in probability:
$    \frac{1}{B}\sum_{t=1}^{B}\phi(X_t)\phi(X_t)^\top \;\xrightarrow[]{p}\; \Sigma_\pi,
    \qquad \Sigma_\pi \succ 0.
$
Moreover, the conditional variances stabilize:
$
    \frac{1}{B}\sum_{t=1}^{B}
    \mathbb{E}\!\left[\big(\widetilde{Y}_t-\tau(X_t)\big)^2\phi(X_t)\phi(X_t)^\top \,\middle|\,\mathcal{F}_{t-1}\right]
    \;\xrightarrow[]{p}\; \Omega_\pi,
$
and a Lindeberg condition holds for the martingale differences $\{\widetilde{Y}_t-\tau(X_t)\}_{t\ge 1}$.
\end{assumption}

\begin{theorem}[Asymptotic normality (martingale CLT)]
\label{thm:clt}
Suppose Assumptions~\ref{assump:predictable}--\ref{assump:stabilize} hold and $\lambda=0$.
Then
$   \sqrt{B}\big(\widehat{\theta}_0-\theta_\star\big)
    \;\xRightarrow[]{d}\;
    \mathcal{N}\!\left(0,\;\Sigma_\pi^{-1}\Omega_\pi\Sigma_\pi^{-1}\right).$
Consequently, for any fixed $x\in\mathcal{X}$,
$
    \sqrt{B}\big(\widehat{\tau}_0(x)-\tau(x)\big)
    \;\xRightarrow[]{d}\;
    \mathcal{N}\!\left(0,\;\phi(x)^\top\Sigma_\pi^{-1}\Omega_\pi\Sigma_\pi^{-1}\phi(x)\right).
$
\end{theorem}

Theorem~\ref{thm:clt} says: \emph{even though the RCT data are collected adaptively and are non-i.i.d.,}
the final estimator still admits classical $\sqrt{B}$-asymptotics.
In practice, this supports uncertainty quantification:
one can estimate the sandwich variance $\Sigma_\pi^{-1}\Omega_\pi\Sigma_\pi^{-1}$ from data and form approximate confidence intervals
for $\tau(x)$ at business-critical segments (e.g., ``new users'' or ``high-value users'').

Notably, the technical challenge for constructing CLT beyond i.i.d is the martingale structure created by randomization:
$\widetilde{Y}_t-\tau(X_t)$ is conditionally mean-zero given the past.
Assumption~\ref{assump:stabilize} requires that the design does not degenerate (the information matrix stabilizes),
which is exactly the failure mode active learning must avoid.
This also clarifies why our acquisition score includes representativeness/shift terms ($d_u$):
they discourage pathological designs where $\Sigma_\pi$ becomes nearly singular.

\subsection{Minimax lower bound and (near) optimality}
\label{sec:theory:minimax}

Upper bounds alone do not tell us whether a method is ``the best possible.''
We therefore complement them with a minimax lower bound.
The message is intentionally sharp:
\emph{even with perfect adaptivity and access to unlimited observational logs and unlabeled pools, one cannot beat the $\sqrt{d/B}$ rate in general.} Consider the linear class
$\mathcal{F}_{\mathrm{lin}}(S)
    \;\triangleq\;
    \left\{\tau_\theta(x)=\langle \theta,\phi(x)\rangle:\ \theta\in\mathbb{R}^d,\ \|\theta\|_2\le S\right\}.$
Assume $\mathbb{P}_X$ puts equal mass on $d$ covariate ``types'' $\{x^{(1)},\dots,x^{(d)}\}$ such that $\phi(x^{(j)})=e_j$ (the $j$-th standard basis).
This captures a realistic situation where the population is a mixture of $d$ distinct segments and CATE differs by segment.

\begin{theorem}[Minimax lower bound for active RCT under bounded randomization]
\label{thm:minimax}
Under the hard instance described above, there exists a universal constant $c>0$ such that for any
adaptive querying policy $\pi$ (possibly using $\mathcal{D}_{\mathrm{obs}}$ and $\mathcal{D}_{\mathrm{pool}}$) and any estimator $\widehat{\tau}$ based on $B$ RCT samples,$
    \inf_{\pi}\ \inf_{\widehat{\tau}}\ \sup_{\tau\in\mathcal{F}_{\mathrm{lin}}(S)}
    \mathbb{E}\!\left[\mathcal{R}(\widehat{\tau})\right]
    \ge
    c\cdot \frac{1}{f_{\min}(1-f_{\max})}\sqrt{\frac{d}{B}}.
$
\end{theorem}

Theorem~\ref{thm:minimax} formalizes a simple but important reality:
if there are $d$ independent degrees of freedom in $\tau(\cdot)$ (e.g., $d$ user segments with genuinely different causal responses),
then with budget $B$ you cannot estimate all of them faster than order $\sqrt{d/B}$ in PEHE.
Active learning can \emph{reallocate} samples to reduce constants,
but it cannot create information out of thin air. It does \emph{not} say that active sampling is useless.
It says that \emph{the best possible scaling} in $B$ is $1/\sqrt{B}$ for this class.
Active sampling is still valuable because it improves the information matrix (hence constants) and prevents degeneracy,
especially when OBS coverage is heavily imbalanced.


\begin{corollary}[Near-minimax optimality of the orthogonalized estimator]
\label{cor:optimality}
Under Assumptions~\ref{assump:predictable}--\ref{assump:linear} and a well-conditioned design $V_0\succeq \kappa B I_d$,
the upper bound in~\eqref{eq:rate_simple} matches the minimax lower bound up to logarithmic factors and the conditioning constant $\kappa$.
In this sense, the estimator in Algorithm~\ref{alg:theory_est} is minimax-rate optimal for $\mathcal{F}_{\mathrm{lin}}(S)$.
\end{corollary}

Also, we justify using $f_k(x)\equiv 1/2$ as a default in Algorithm~\ref{alg:active_fusion} could minimize the conditional variance,
and using clipping $[f_{\min},f_{\max}]$ as a principled safeguard when $p$ must vary by covariates.

%% file: chapters/exp.tex
\section{Experiments}\label{sec:experiment}

\textbf{Datasets.} Our empirical evaluation utilizes a large-scale \emph{real-world dataset} derived from a ride-hailing platform. 
The dataset is collected at the order level, with each instance corresponding to an individual order characterized by $468$ features. Formally, let $T \in \{0, 1\}$ denote the treatment arm, indicating whether a ``free upgrade'' service was triggered for the passenger.
Correspondingly, the outcome $Y \in \{0, 1\}$ is defined as the order completion status. We implement an \textbf{Out-of-Time (OOT) protocol} spanning 44 days. During the initial 30-day phase, we establish the RCT Pool $\mathcal{D}_{\mathrm{pool}}$ and concurrently collect observational samples to form the training set (either $\mathcal{D}_{\mathrm{obs}}^{\mathrm{bias}}$ or $\mathcal{D}_{\mathrm{obs}}^{\mathrm{full}}$). Generalization is strictly evaluated on the Unbiased RCT Test Set $\mathcal{D}_{\mathrm{rct}}^{\mathrm{test}}$, which comprises RCT data collected during the remaining 14 days. Detailed descriptions are provided in Appendix~\ref{sec:data_details}.

\textbf{Baselines \& Metrics. } In our experiments, we evaluate three industrial deep uplift baselines—DragonNet \citep{shi2019adapting}, DESCN \citep{zhong2022descn}, and DRCFR \citep{cheng2022learning}—under active versus random sampling. Due to the unobservability of counterfactuals in real-world industrial RCTs, direct calculation of PEHE is infeasible. Instead, we employ the Normalized Area Under the Uplift Curve (AUUC) \citep{pmlr-v67-gutierrez17a}, a practically-oriented metric widely adopted in industry to assess ranking quality.

\textbf{Main Results.} We demonstrate that the proposed active sampling strategy significantly enhances sample efficiency in hybrid uplift modeling, maximizing the marginal utility of limited experimental data. As illustrated in Figure~\ref{fig:active_base_model}, our approach yields consistent performance gains across all evaluated architectures compared to random sampling. Specifically, the \textbf{DRCFR} model coupled with active sampling achieves an AUUC at $100\mathrm{k}$ samples comparable to that of random sampling at $500\mathrm{k}$, effectively reducing the RCT labeling cost by approximately $80\%$. For other backbones, \textbf{DESCN} performs best at $10\mathrm{k}$, while \textbf{DragonNet} peaks at $500\mathrm{k}$. Notably, active sampling stabilizes DragonNet by preventing the large fluctuations observed under random sampling. Supplementary results regarding data scalability, comprehensive ablation studies, and detailed experimental protocols are provided in Appendix~\ref{sec:additional_exp}.

\begin{figure}[t]
    \centering
    \includegraphics[width=0.98\linewidth]{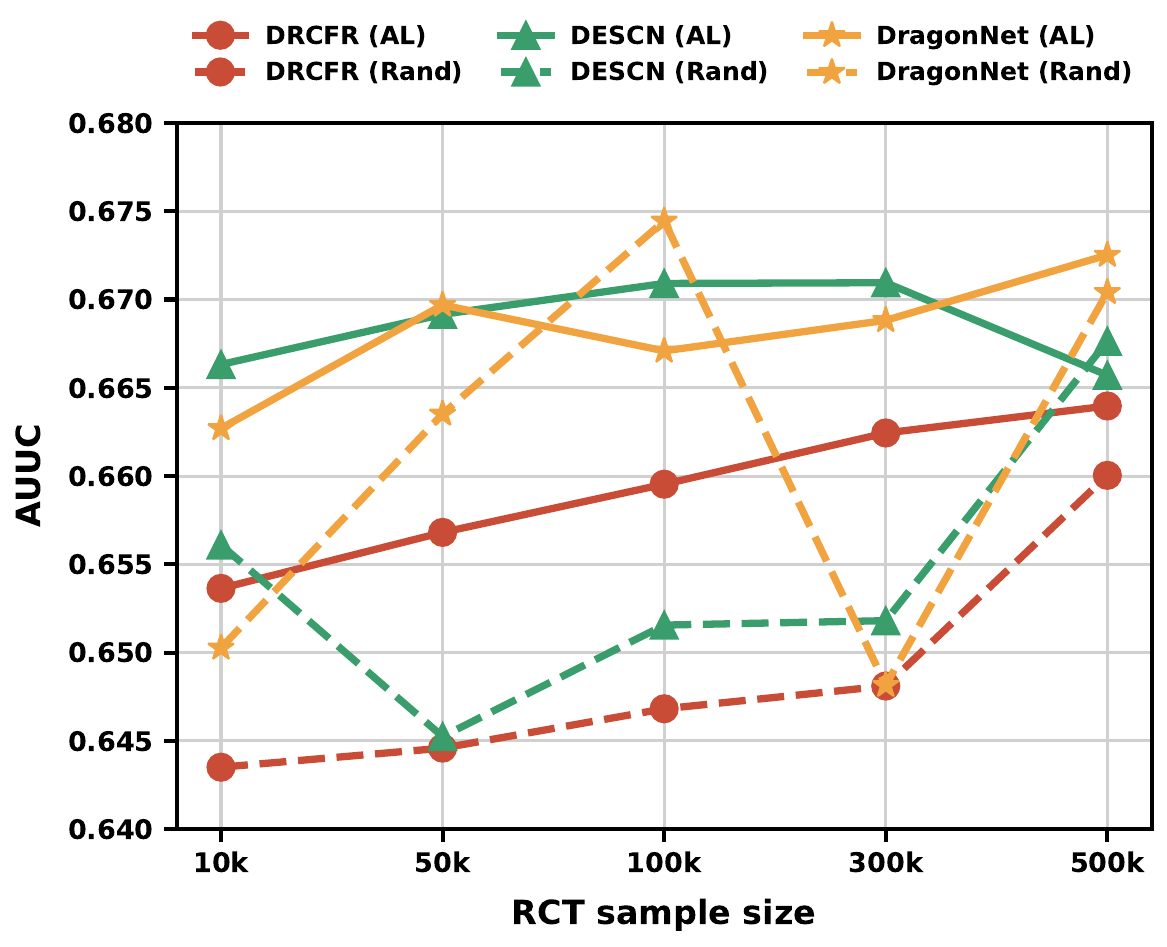}
    \caption{Performance comparison of Active Learning (AL) versus Random Sampling (Rand) strategies across varying RCT sample sizes (10k to 500k). The curves display the AUUC scores for DRCFR, DESCN, and DragonNet models. The experiments are conducted on the Biased Observational Set $\mathcal{D}_{\text{obs}}^{\text{bias}}$, demonstrating the superior data efficiency of active sampling (solid lines) compared to random sampling (dashed lines) under distribution shifts.}
    \label{fig:active_base_model}
\end{figure}

%% file: chapters/conclu.tex
\section{Conclusion}\label{sec:conclusion}


We study heterogeneous treatment effect estimation under a pervasive real-world constraint: randomized controlled trials are scarce and costly, while observational data are abundant but biased. 
In our framework, observational data inform \emph{where} randomized experimentation should occur, while causal identification is carried exclusively by randomized experiments. We formalize this perspective as a budgeted active experimentation problem
and propose a multi-criteria acquisition strategy that leverages observational data
to diagnose uncertainty, overlap deficits, and covariate shift,
thereby directing limited RCT budget to the most informative regions.
We show that randomization induces a martingale structure that guarantees unbiasedness,
admits finite-sample deviation bounds, and yields asymptotic normality.
Moreover, we establish minimax lower bounds demonstrating that the achieved $\sqrt{d/B}$ rate is information-theoretically optimal, clarifying that active learning improves efficiency through better experimental design rather than by circumventing fundamental limits.

Furthermore, extending the theoretical guarantees beyond linear realizability to richer nonparametric or deep function classes remains an important challenge, and jointly optimizing unit selection and treatment randomization probabilities may further improve efficiency under operational constraints.
We believe these directions will further strengthen active experimentation
as a foundational paradigm for cost-efficient causal learning.

%% file: chapters/appendix.tex
\clearpage
\appendix
\onecolumn

\section{Related Work}\label{sec:related}
Industrial CATE estimation often faces challenge: limited RCT budgets versus abundant but biased observational data.
This motivates us to develop an OBS \& RCT data fusion framework that casts experiment design as a budgeted active learning problem. In this section, we discuss related work on (i) \emph{active sampling for budget-limited experiments}, (ii) \emph{active sampling in observational studies}, and (iii) \emph{active learning with neural networks}.

\textbf{(i) Active sampling for budget-limited experiments.}
In sample-constrained RCT settings, the unlabeled pool is available but only a small subset of individuals can be enrolled, and the goal is to minimize estimation error via selective recruitment. \citet{addanki2022sample} utilize leverage-based selection and balanced assignment to optimize ATE deviation under fixed budgets. Similarly, \citet{zhangactive} propose adaptive sampling with reweighting (RWAS) to achieve stable, sample-efficient ATE estimation within a limited draw budget.

\textbf{(ii) Active sampling in observational studies.}
In observational settings, treatment is pre-assigned, and the budget restricts label queries. \citet{wen2025enhancingtreatmenteffectestimation} prioritize enhancing \emph{factual and counterfactual coverage} to improve overlap and reduce PEHE. Alternatively, Causal-EPIG \citep{gao2025causalepigpredictionorientedactivelearning} employs an information-theoretic approach, selecting samples that maximize expected predictive information gain to minimize CATE uncertainty under a limited labeling budget.

\textbf{(iii) Active learning with neural networks.}
In neural settings, acquisition is often driven by disagreement or uncertainty. \citet{zhu2022active} propose NeuralCAL, which queries points in the \emph{disagreement region} of plausible predictors to reduce label complexity. Its extension, NeuralCAL++, further integrates \emph{abstention} mechanisms and confidence intervals to reach target error rates with significantly fewer queries.

\begin{table}[htbp]
    \centering
    \small 
    \caption{Comparison of related works in Active Learning (AL) for Treatment Effect Estimation.}
    \label{tab:related_work}
    \renewcommand{\arraystretch}{1.3}
    
    \begin{tabularx}{\textwidth}{l l X c X l}
        \toprule
        \textbf{AL Methods} & \textbf{Data} & \textbf{\makecell{Objective \\ \& Budget}} & \textbf{\makecell{Guides Exp. \\ Design?}} & \textbf{\makecell{Acquisition \\ Strategy}} & \textbf{Evaluation Metric} \\
        \midrule

        \makecell[l]{Ours.\\ \textbf{(AL for RCT\&OBS TE)}}& 
        \makecell[l]{OBS + RCT} & 
        \textbf{Obj:} Min. CATE Error \newline \textbf{Budget:} RCT Experiment Budget & 
        \cmark & 
        \textbf{Ours:} Top-$k$ by uncertainty + discrepancy + overlap deficit. & 
        \textbf{AUUC} (For CATE) \\
        \midrule
        \midrule
        
        \makecell[l]{Sample-Constrained TE Est.\\ \cite{addanki2022sample}} & 
        RCT only& 
        \textbf{Obj:} Min. ATE/ITE Error \newline \textbf{Budget:} Limited Sample Size & 
        \cmark & 
        \textbf{SCTE:} Top-$k$ informative units with covariate-balanced assignment. & 
        \makecell[t]{\textbf{RMSE} (For ITE) \\ \textbf{Deviation} (For ATE)} \\
        \midrule
        
        \makecell[l]{Active TE Est.\\ \cite{zhangactive}} & 
        RCT only& 
        \textbf{Obj:} Min. ATE Error \newline \textbf{Budget:} Limited Sample Size & 
        \cmark & 
        \textbf{RWAS:} Probability sampling + inverse-propensity reweighting. & 
        \textbf{Deviation} (For ATE) \\
        \midrule

        \makecell[l]{Enhancing TE Est.\\ \cite{wen2025enhancingtreatmenteffectestimation}} & 
        OBS only& 
        \textbf{Obj:} Min. CATE Error \newline \textbf{Budget:} Labeling Budget & 
        \xmark & 
        \textbf{Coverage:} Top-$k$ maximizing factual/counterfactual coverage. & 
        \makecell[t]{\textbf{RMSE} (For ITE) \\ \textbf{PEHE} (For CATE)} \\
        \midrule

        \makecell[l]{Causal-EPIG\\ \cite{gao2025causalepigpredictionorientedactivelearning}} & 
        OBS only& 
        \textbf{Obj:} Min. CATE Error \newline \textbf{Budget:} Labeling Budget & 
        \xmark & 
        \textbf{Causal-EPIG:} Top-$k$ by expected predictive information gain. & 
        \textbf{PEHE} (For CATE) \\
        \midrule

        \makecell[l]{AL with Neural Networks\\\cite{zhu2022active}} & 
        Supervised & 
        \textbf{Obj:} Optim. Label Complexity and Excess Error \newline \textbf{Budget:} Query Count & 
        \xmark & 
        \textbf{NeuralCAL(++):} Top-$k$ by disagreement/ambiguity (with abstention). & 
        \makecell[t]{\textbf{Excess Error}\\ \textbf{Label Complexity}} \\
        \bottomrule
    \end{tabularx}
    
    \footnotesize
    \textbf{Note:} OBS = Observational Data, RCT = Randomized Controlled Trial. \textbf{Guides Exp. Design?}: Indicates if the method actively assigns treatments (\cmark) or only queries labels (\xmark).
\end{table}

\section{Theoretical Analysis}
\label{sec:theory_appendix}

This section provides theoretical guarantees for the \emph{estimation component} of Algorithm~\ref{alg:active_fusion}
under the adaptively collected RCT data.
Our goal is to make three claims precise:
(i) \textbf{Finite-sample validity}: an (near-) unbiased effect estimator and a non-asymptotic error bound that remains valid under \emph{adaptive, non-i.i.d.} RCT sampling.
(ii) \textbf{Asymptotic inference}: a central limit theorem (CLT) showing asymptotic normality enabling confidence intervals.
(iii) \textbf{Fundamental limits}: a minimax lower bound showing that the achieved rate is information-theoretically optimal (up to logs).
A key theme is that \emph{randomization protects unbiasedness, while active selection shapes the information matrix and hence the error rate.}

We first formalize an analyzable estimator that matches the RCT protocol in Algorithm~\ref{alg:active_fusion} and
is standard in modern causal ML: regress an \emph{orthogonalized pseudo-outcome} onto a representation.

\subsection{Setup: adaptive RCT stream and pseudo-outcome regression}
\label{sec:theory:setup}

Let $B$ denote the total RCT budget and index queried units in chronological order $t=1,\dots,B$.
Write the $t$-th queried sample as $(X_t,T_t,Y_t,p_t)$ where $p_t\in[f_{\min},f_{\max}]$ is the (known) assignment probability
used when unit $t$ was experimented.
Let $\{\mathcal{F}_{t}\}_{t\ge 0}$ be the filtration where $\mathcal{F}_{t}$ contains
$\mathcal{D}_{\mathrm{obs}}$, the full covariate pool $\mathcal{D}_{\mathrm{pool}}$, and all RCT data up to time $t$.
The adaptive querying policy (via $\mathcal{C}_k$, $f_k$) implies that $X_t$ and $p_t$ may depend on $\mathcal{F}_{t-1}$.

We define the RCT pseudo-outcome
\begin{equation}
    \widetilde{Y}_t
    \;\triangleq\;
    \frac{T_t\,Y_t}{p_t}
    \;-\;
    \frac{(1-T_t)\,Y_t}{1-p_t}.
\end{equation}
Intuitively, $\widetilde{Y}_t$ converts each single randomized trial into a noisy but unbiased ``label'' of the treatment effect at $X_t$.

Let $\phi(x)\in\mathbb{R}^d$ be a fixed feature map (e.g., the learned embedding $\phi$ in Algorithm~\ref{alg:active_fusion}).
For theory, we analyze a realizable linear CATE model in this representation:
\begin{equation}
    \tau(x) \;=\; \langle \theta_\star, \phi(x)\rangle, \qquad \theta_\star\in\mathbb{R}^d.
\end{equation}
Given $\{(X_t,T_t,Y_t,p_t)\}_{t=1}^B$, we estimate $\theta_\star$ by (optionally regularized) least squares on pseudo-outcomes:
\begin{equation}
    \widehat{\theta}_\lambda
    \;\triangleq\;
    \arg\min_{\theta\in\mathbb{R}^d}\;
    \sum_{t=1}^{B}\Big(\widetilde{Y}_t-\langle \theta,\phi(X_t)\rangle\Big)^2
    \;+\;\lambda\|\theta\|_2^2,
\end{equation}
and output $\widehat{\tau}_\lambda(x)\triangleq \langle \widehat{\theta}_\lambda,\phi(x)\rangle$.
Algorithm~\ref{alg:theory_est} summarizes this estimator.

Algorithm~\ref{alg:active_fusion} specifies \emph{how the $X_t$'s are chosen} (active sampling) and
\emph{how $p_t$ is assigned} (bounded randomization).
Algorithm~\ref{alg:theory_est} isolates a transparent estimator for analyzing the statistical consequences of that adaptive design.
In practice, one may replace the linear regressor by a neural CATE learner; the key objects in the proofs below are
(i) the unbiased pseudo-outcome, and
(ii) an information matrix that quantifies how well the queried points cover the representation space.

\subsection{Unbiasedness and finite-sample error bounds}

We now state the first main result: despite \emph{adaptive} (non-i.i.d.) selection of $X_t$,
the RCT pseudo-outcome remains unbiased, which implies unbiasedness of the linear estimator (for $\lambda=0$)
and a finite-sample deviation bound (for $\lambda>0$).

\begin{assumption}{(Predictable adaptive design).}
For each $t$, the queried covariate $X_t$ and assignment probability $p_t$ are $\mathcal{F}_{t-1}$-measurable.
\end{assumption}

\begin{assumption}{(RCT randomization and boundedness).}
Conditioned on $(X_t,p_t,\mathcal{F}_{t-1})$, treatment is randomized as
$T_t \sim \mathrm{Bern}(p_t)$ and $T_t \perp (Y_t(0),Y_t(1)) \mid (X_t,p_t,\mathcal{F}_{t-1})$.
Outcomes are bounded: $Y_t(0),Y_t(1)\in[0,1]$ almost surely, and $p_t\in[f_{\min},f_{\max}]$ with $0<f_{\min}\le f_{\max}<1$.
\end{assumption}

\begin{assumption}{(Linear realizability in representation space).}
There exists $\theta_\star\in\mathbb{R}^d$ such that $\tau(x)=\langle \theta_\star,\phi(x)\rangle$ for all $x$.
Moreover, $\|\theta_\star\|_2\le S$ and $\|\phi(x)\|_2\le L$ for all $x$.
\end{assumption}

\begin{lemma}[Unbiased pseudo-outcome under adaptive sampling]
Under Assumptions~\ref{assump:predictable}--\ref{assump:rct_bound},
\begin{equation}
    \mathbb{E}\!\left[\widetilde{Y}_t \mid X_t,p_t,\mathcal{F}_{t-1}\right] \;=\; \tau(X_t),
    \qquad t=1,\dots,B.
\end{equation}
Moreover, $\widetilde{Y}_t$ is uniformly bounded:
\begin{equation}
    \label{eq:pseudo_bound}
    |\widetilde{Y}_t|
    \;\le\;
    \max\!\left\{\frac{1}{f_{\min}},\,\frac{1}{1-f_{\max}}\right\}
    \;\triangleq\; L_p.
\end{equation}
\end{lemma}

\begin{proof}
Fix $t\in\{1,\dots,B\}$.
By consistency and $Y_t = T_t Y_t(1) + (1-T_t) Y_t(0)$,
\[
T_t Y_t = T_t Y_t(1),\qquad (1-T_t)Y_t = (1-T_t)Y_t(0).
\]
Hence,
\[
\widetilde{Y}_t
=
\frac{T_t Y_t(1)}{p_t}
-
\frac{(1-T_t) Y_t(0)}{1-p_t}.
\]
Taking conditional expectation given $(X_t,p_t,\mathcal{F}_{t-1})$ and using Assumption~\ref{assump:rct_bound},
\begin{align*}
\mathbb{E}\!\left[\frac{T_t Y_t(1)}{p_t}\,\middle|\,X_t,p_t,\mathcal{F}_{t-1}\right]
&=
\mathbb{E}\!\left[\,
\mathbb{E}\!\left[\frac{T_t}{p_t}\,\middle|\,X_t,p_t,\mathcal{F}_{t-1},Y_t(1)\right]\,Y_t(1)
\,\middle|\,X_t,p_t,\mathcal{F}_{t-1}
\right] \\
&=
\mathbb{E}\!\left[\,
\frac{\mathbb{E}[T_t\mid X_t,p_t,\mathcal{F}_{t-1}]}{p_t}\,Y_t(1)
\,\middle|\,X_t,p_t,\mathcal{F}_{t-1}
\right] \\
&=
\mathbb{E}\!\left[\,Y_t(1)\mid X_t,p_t,\mathcal{F}_{t-1}\right],
\end{align*}
where the second line uses conditional independence $T_t \perp Y_t(1)\mid (X_t,p_t,\mathcal{F}_{t-1})$ and the last line uses
$\mathbb{E}[T_t\mid X_t,p_t,\mathcal{F}_{t-1}]=p_t$.
Similarly,
\[
\mathbb{E}\!\left[\frac{(1-T_t) Y_t(0)}{1-p_t}\,\middle|\,X_t,p_t,\mathcal{F}_{t-1}\right]
=
\mathbb{E}\!\left[\,Y_t(0)\mid X_t,p_t,\mathcal{F}_{t-1}\right].
\]
Therefore,
\[
\mathbb{E}\!\left[\widetilde{Y}_t\mid X_t,p_t,\mathcal{F}_{t-1}\right]
=
\mathbb{E}\!\left[Y_t(1)-Y_t(0)\mid X_t,p_t,\mathcal{F}_{t-1}\right]
=
\mathbb{E}\!\left[Y_t(1)-Y_t(0)\mid X_t\right]
=\tau(X_t),
\]
where the second equality uses the randomization/ignorability in Assumption~\ref{assump:rct_bound} and the last equality is the definition of CATE.

For boundedness, note that exactly one of the two terms in the pseudo outcome is nonzero:
if $T_t=1$, then $\widetilde{Y}_t = Y_t/p_t \in [0,1/p_t]$;
if $T_t=0$, then $\widetilde{Y}_t = -Y_t/(1-p_t)\in[-1/(1-p_t),0]$.
Using $p_t\ge f_{\min}$ and $1-p_t\ge 1-f_{\max}$ gives~\eqref{eq:pseudo_bound}.
\end{proof}

Equation~\eqref{eq:unbiased_pseudo} says: \emph{even if we cherry-pick who enters the experiment},
as long as we randomize treatment for the selected unit, each queried unit yields an unbiased, single-sample estimate of its own CATE.
A concrete example is coupon delivery:
even if we actively choose ``hard'' users (e.g., high-value users with extreme historical targeting),
a randomized coupon/no-coupon decision still produces an unbiased causal contrast at that user.
This lemma is \emph{not} claiming that the observational estimator is unbiased, nor that any neural training procedure is unbiased.
It only uses RCT randomization and boundedness.
The adaptive selection affects \emph{where} we learn (which $X_t$'s), not \emph{whether} a queried unit is causally valid.

A potential skeptic is that ``the learner chooses $X_t$ based on past outcomes; isn't this a form of selection bias that breaks unbiasedness?''
The resolution is that unbiasedness in~\eqref{eq:unbiased_pseudo} is \emph{conditional} on $(X_t,p_t,\mathcal{F}_{t-1})$:
the selection may change the distribution of $X_t$, but $T_t$ is randomized \emph{after} selection and is conditionally independent of potential outcomes.
Technically, $\widetilde{Y}_t-\tau(X_t)$ forms a martingale difference sequence.

\paragraph{A key martingale object.}
Define the centered pseudo-outcome noise
\begin{equation}
    \label{eq:eps_def}
    \varepsilon_t \;\triangleq\; \widetilde{Y}_t - \tau(X_t),
\end{equation}
and the (predictable) feature vector $\phi_t\triangleq \phi(X_t)$.
By Lemma~\ref{lem:unbiased_pseudo}, $\mathbb{E}[\varepsilon_t\mid \mathcal{F}_{t-1},X_t,p_t]=0$ and $|\varepsilon_t|\le 2L_p$ (since both $\widetilde{Y}_t$ and $\tau(X_t)\in[-1,1]$ are bounded).
Let
\begin{equation}
    \label{eq:V_def}
    V_\lambda \;\triangleq\; \lambda I_d + \sum_{t=1}^{B}\phi_t\phi_t^\top, \qquad
    M_B \;\triangleq\; \sum_{t=1}^{B}\phi_t\,\varepsilon_t.
\end{equation}
Then $\{M_t\}_{t\ge 0}$ with $M_t=\sum_{s=1}^t \phi_s\varepsilon_s$ is a vector-valued martingale with respect to $\{\mathcal{F}_t\}$.

\paragraph{A self-normalized concentration lemma (proved in full).}
The finite-sample bound relies on a vector-valued self-normalized martingale inequality.
Because this step is often cited as a black box, we provide a complete proof to make the non-i.i.d. aspect fully transparent.

\begin{lemma}[Conditional sub-Gaussianity of $\varepsilon_t$]
\label{lem:subg}
Under Assumptions~\ref{assump:predictable}--\ref{assump:rct_bound}, for each $t$ and any $\lambda\in\mathbb{R}$,
\begin{equation}
    \label{eq:subg}
\mathbb{E}\!\left[\exp\!\big(\lambda\,\varepsilon_t\big)\,\middle|\,\mathcal{F}_{t-1},X_t,p_t\right]
    \;\le\; \exp\!\left(\frac{\lambda^2\sigma^2}{2}\right),
    \qquad \text{with }\ \sigma \triangleq 2L_p.
\end{equation}
\end{lemma}

\begin{proof}
Condition on $(\mathcal{F}_{t-1},X_t,p_t)$.
By Lemma~\ref{lem:unbiased_pseudo}, $\mathbb{E}[\varepsilon_t\mid \mathcal{F}_{t-1},X_t,p_t]=0$.
Moreover, $Y_t\in[0,1]$ and $p_t\in[f_{\min},f_{\max}]$ imply $|\widetilde{Y}_t|\le L_p$ by~\eqref{eq:pseudo_bound}.
Also $\tau(X_t)=\mathbb{E}[Y(1)-Y(0)\mid X_t]\in[-1,1]$ because $Y(1),Y(0)\in[0,1]$.
Hence $|\varepsilon_t| \le |\widetilde{Y}_t|+|\tau(X_t)| \le L_p+1 \le 2L_p$ since $L_p\ge 1$.

Now apply Hoeffding's lemma in conditional form:
if $Z$ is zero-mean and almost surely lies in $[a,b]$ given a sigma-field, then
$\mathbb{E}[e^{\lambda Z}\mid \cdot] \le \exp(\lambda^2(b-a)^2/8)$.
Here $Z=\varepsilon_t$ and $[a,b]=[-2L_p,2L_p]$, so $(b-a)^2/8=(4L_p)^2/8=2L_p^2$.
Thus~\eqref{eq:subg} holds with $\sigma=2L_p$.
\end{proof}

\begin{lemma}[Self-normalized martingale inequality (vector form)]
\label{lem:selfnorm}
Assume $\lambda>0$ and let $V_\lambda$ and $M_B$ be defined in~\eqref{eq:V_def}.
Suppose $\varepsilon_t$ satisfies the conditional sub-Gaussian property~\eqref{eq:subg} with parameter $\sigma$ and $\phi_t$ is $\mathcal{F}_{t-1}$-measurable.
Then for any $\delta\in(0,1)$, with probability at least $1-\delta$,
\begin{equation}
    \label{eq:selfnorm_lemma}
    \|M_B\|_{V_\lambda^{-1}}
    \;\le\;
    \sigma\sqrt{
        2\log\!\left(
            \frac{\det(V_\lambda)^{1/2}}{\det(\lambda I_d)^{1/2}\,\delta}
        \right)
    }.
\end{equation}
Here $\|z\|_{A}\triangleq \sqrt{z^\top A z}$ for any positive semidefinite $A$.
\end{lemma}

\begin{proof}
The proof proceeds by constructing an exponential supermartingale and applying a mixture (Gaussian integration) argument.

\textbf{Step 1: exponential supermartingale for a fixed direction.}
Fix any $u\in\mathbb{R}^d$ and define for $t=0,1,\dots,B$,
\begin{equation}
    \label{eq:Ztu}
    Z_t(u)
    \;\triangleq\;
    \exp\!\left(
        \frac{1}{\sigma}u^\top M_t
        \;-\;
        \frac{1}{2}\,u^\top\Big(\sum_{s=1}^{t}\phi_s\phi_s^\top\Big)u
    \right),
\end{equation}
with the convention $M_0=0$ and the empty sum equals $0$.

We claim that $\{Z_t(u)\}_{t=0}^B$ is a nonnegative supermartingale w.r.t.\ $\{\mathcal{F}_t\}$.
Indeed, using $M_t=M_{t-1}+\phi_t\varepsilon_t$ and $\phi_t$ being $\mathcal{F}_{t-1}$-measurable,
\begin{align*}
Z_t(u)
&=
Z_{t-1}(u)\cdot
\exp\!\left(
        \frac{1}{\sigma}u^\top\phi_t\,\varepsilon_t
        -
        \frac{1}{2}u^\top\phi_t\phi_t^\top u
\right) \\
&=
Z_{t-1}(u)\cdot
\exp\!\left(
        \frac{1}{\sigma}(u^\top\phi_t)\,\varepsilon_t
        -
        \frac{1}{2}(u^\top\phi_t)^2
\right).
\end{align*}
Taking conditional expectation given $\mathcal{F}_{t-1}$ and applying~\eqref{eq:subg} with $\lambda=(u^\top\phi_t)/\sigma$ yields
\begin{align*}
\mathbb{E}\!\left[ Z_t(u)\mid \mathcal{F}_{t-1}\right]
&=
Z_{t-1}(u)\cdot
\exp\!\left(-\frac{1}{2}(u^\top\phi_t)^2\right)
\cdot
\mathbb{E}\!\left[\exp\!\left(\frac{1}{\sigma}(u^\top\phi_t)\varepsilon_t\right)\,\middle|\,\mathcal{F}_{t-1}\right] \\
&\le
Z_{t-1}(u)\cdot
\exp\!\left(-\frac{1}{2}(u^\top\phi_t)^2\right)
\cdot
\exp\!\left(\frac{1}{2}(u^\top\phi_t)^2\right)
=
Z_{t-1}(u).
\end{align*}
Thus $\{Z_t(u)\}$ is a supermartingale and in particular $\mathbb{E}[Z_B(u)]\le Z_0(u)=1$.

\textbf{Step 2: mixture over $u$ to obtain a uniform bound.}
Let $U\sim\mathcal{N}(0,\lambda^{-1}I_d)$ be independent of everything else.
Define the mixture random variable
\[
\overline Z_B \;\triangleq\; \mathbb{E}\!\left[ Z_B(U)\,\middle|\,\mathcal{F}_B\right].
\]
Since $Z_B(u)\ge 0$ for all $u$, by Fubini's theorem and the supermartingale property,
\[
\mathbb{E}[\overline Z_B] = \mathbb{E}\big[\mathbb{E}[Z_B(U)\mid \mathcal{F}_B]\big] = \mathbb{E}[Z_B(U)] = \mathbb{E}\big[\mathbb{E}[Z_B(U)]\big] \le 1.
\]

\textbf{Step 3: explicit evaluation of the Gaussian integral.}
Write $\Sigma_B\triangleq \sum_{t=1}^{B}\phi_t\phi_t^\top$ so that $V_\lambda=\lambda I_d+\Sigma_B$.
Using the definition~\eqref{eq:Ztu} at $t=B$,
\[
Z_B(u) = \exp\!\left(\frac{1}{\sigma}u^\top M_B - \frac{1}{2}u^\top \Sigma_B u\right).
\]
The density of $U\sim\mathcal{N}(0,\lambda^{-1}I_d)$ is
\[
(2\pi)^{-d/2}\det(\lambda^{-1}I_d)^{-1/2}\exp\!\left(-\frac{1}{2}u^\top(\lambda I_d)u\right)
=
(2\pi)^{-d/2}\det(\lambda I_d)^{1/2}\exp\!\left(-\frac{1}{2}u^\top(\lambda I_d)u\right).
\]
Therefore, conditioning on $\mathcal{F}_B$ (so $M_B$ and $\Sigma_B$ are fixed),
\begin{align*}
\overline Z_B
&=
\int_{\mathbb{R}^d}
\exp\!\left(\frac{1}{\sigma}u^\top M_B - \frac{1}{2}u^\top \Sigma_B u\right)
\cdot
(2\pi)^{-d/2}\det(\lambda I_d)^{1/2}\exp\!\left(-\frac{1}{2}u^\top(\lambda I_d)u\right)\,du \\
&=
(2\pi)^{-d/2}\det(\lambda I_d)^{1/2}
\int_{\mathbb{R}^d}
\exp\!\left(\frac{1}{\sigma}u^\top M_B - \frac{1}{2}u^\top (\Sigma_B+\lambda I_d)u\right)\,du \\
&=
(2\pi)^{-d/2}\det(\lambda I_d)^{1/2}
\int_{\mathbb{R}^d}
\exp\!\left(\frac{1}{\sigma}u^\top M_B - \frac{1}{2}u^\top V_\lambda u\right)\,du.
\end{align*}
Applying the standard Gaussian integral identity
\[
\int_{\mathbb{R}^d}\exp\!\left(b^\top u - \frac{1}{2}u^\top A u\right)\,du
=
(2\pi)^{d/2}\det(A)^{-1/2}\exp\!\left(\frac{1}{2}b^\top A^{-1}b\right),
\quad A\succ 0,
\]
with $A=V_\lambda$ and $b=M_B/\sigma$ gives
\begin{equation}
    \label{eq:mixture_eval}
    \overline Z_B
    =
    \frac{\det(\lambda I_d)^{1/2}}{\det(V_\lambda)^{1/2}}
    \exp\!\left(\frac{1}{2\sigma^2}M_B^\top V_\lambda^{-1}M_B\right)
    =
    \frac{\det(\lambda I_d)^{1/2}}{\det(V_\lambda)^{1/2}}
    \exp\!\left(\frac{1}{2\sigma^2}\|M_B\|_{V_\lambda^{-1}}^2\right).
\end{equation}

\textbf{Step 4: apply Markov's inequality.}
Since $\mathbb{E}[\overline Z_B]\le 1$, for any $\delta\in(0,1)$,
\[
\mathbb{P}\!\left(\overline Z_B \ge \frac{1}{\delta}\right) \le \delta.
\]
On the complement event (probability at least $1-\delta$), combine with~\eqref{eq:mixture_eval} to obtain
\[
\frac{\det(\lambda I_d)^{1/2}}{\det(V_\lambda)^{1/2}}
\exp\!\left(\frac{1}{2\sigma^2}\|M_B\|_{V_\lambda^{-1}}^2\right)
\le \frac{1}{\delta}.
\]
Taking logarithms and rearranging yields exactly~\eqref{eq:selfnorm_lemma}.
\end{proof}

\paragraph{Main finite-sample theorem.}
We can now prove the finite-sample guarantee in full detail.

\begin{theorem}[Finite-sample unbiasedness and deviation bound]
Assume~\ref{assump:predictable}--\ref{assump:linear} and define $V_\lambda$ as in~\eqref{eq:V_def}.
Let $\widehat{\theta}_\lambda$ be defined by~\eqref{eq:ridge}.
Then:
\begin{enumerate}
    \item \textbf{(Conditional unbiasedness for OLS).}
    If $\lambda=0$ and $V_0$ is invertible, then
    \begin{equation}
        \label{eq:ols_unbiased}
        \mathbb{E}\!\left[\widehat{\theta}_0 \mid \{X_t,p_t\}_{t=1}^B\right]=\theta_\star,
        \qquad
        \mathbb{E}\!\left[\widehat{\tau}_0(x)\mid \{X_t,p_t\}_{t=1}^B\right]=\tau(x),\ \forall x\in\mathcal{X}.
    \end{equation}
    \item \textbf{(High-probability self-normalized bound).}
    Fix $\delta\in(0,1)$ and any $\lambda>0$.
    With probability at least $1-\delta$,
    \begin{equation}
        \big\|\widehat{\theta}_\lambda-\theta_\star\big\|_{V_\lambda}
        \;\le\;
        \underbrace{\sigma\sqrt{2\log\!\left(\frac{\det(V_\lambda)^{1/2}}{\det(\lambda I_d)^{1/2}\,\delta}\right)}}_{\text{stochastic term}}
        \;+\;
        \underbrace{\sqrt{\lambda}\,S}_{\text{regularization bias}},
    \end{equation}
    where $\sigma$ can be taken as $2L_p$ (Lemma~\ref{lem:subg}).
    Consequently, for every $x\in\mathcal{X}$,
    \begin{equation}
        \big|\widehat{\tau}_\lambda(x)-\tau(x)\big|
        \;\le\;
        \beta_{B}(\delta)\cdot
        \sqrt{\phi(x)^\top V_\lambda^{-1}\phi(x)},
        \qquad
        \beta_{B}(\delta)\triangleq
        \sigma\sqrt{2\log\!\left(\frac{\det(V_\lambda)^{1/2}}{\det(\lambda I_d)^{1/2}\,\delta}\right)}
        +\sqrt{\lambda}S .
    \end{equation}
\end{enumerate}
\end{theorem}

\begin{proof}
\textbf{Part (i): conditional unbiasedness for OLS ($\lambda=0$).}
Let $\Phi\in\mathbb{R}^{B\times d}$ be the design matrix with $t$-th row $\phi_t^\top$.
Let $\widetilde Y\in\mathbb{R}^{B}$ be the vector with entries $\widetilde Y_t$.
By linear realizability, $\tau(X_t)=\phi_t^\top\theta_\star$ for all $t$.
Define the noise vector $\varepsilon\in\mathbb{R}^{B}$ with entries $\varepsilon_t=\widetilde Y_t-\tau(X_t)$.
Then we have the exact decomposition
\begin{equation}
    \label{eq:linear_decomp}
    \widetilde Y = \Phi\theta_\star + \varepsilon.
\end{equation}
By Lemma~\ref{lem:unbiased_pseudo}, $\mathbb{E}[\varepsilon_t\mid X_t,p_t,\mathcal{F}_{t-1}]=0$; in particular,
conditioning on the full design $\{X_t,p_t\}_{t=1}^B$ implies $\mathbb{E}[\varepsilon\mid \{X_t,p_t\}_{t=1}^B]=0$.

When $\lambda=0$ and $V_0=\Phi^\top\Phi$ is invertible, the OLS solution is
\[
\widehat{\theta}_0 = (\Phi^\top\Phi)^{-1}\Phi^\top \widetilde Y
= (\Phi^\top\Phi)^{-1}\Phi^\top(\Phi\theta_\star+\varepsilon)
= \theta_\star + (\Phi^\top\Phi)^{-1}\Phi^\top\varepsilon.
\]
Taking conditional expectation given $\{X_t,p_t\}_{t=1}^B$ yields
$\mathbb{E}[\widehat\theta_0\mid \{X_t,p_t\}]=\theta_\star$, proving the first claim in~\eqref{eq:ols_unbiased}.
The second claim follows since $\widehat\tau_0(x)=\phi(x)^\top\widehat\theta_0$ and $\tau(x)=\phi(x)^\top\theta_\star$.

\textbf{Part (ii): high-probability self-normalized bound ($\lambda>0$).}
From the ridge normal equations, the closed form is
\[
\widehat{\theta}_\lambda
=
V_\lambda^{-1}\sum_{t=1}^B \phi_t \widetilde Y_t
=
V_\lambda^{-1}\sum_{t=1}^B \phi_t(\phi_t^\top\theta_\star+\varepsilon_t)
=
V_\lambda^{-1}\Big(\sum_{t=1}^B\phi_t\phi_t^\top\Big)\theta_\star
+
V_\lambda^{-1}\sum_{t=1}^B \phi_t\varepsilon_t.
\]
Since $\sum_{t=1}^B\phi_t\phi_t^\top = V_\lambda-\lambda I_d$, we get the exact identity
\begin{equation}
    \label{eq:theta_error_decomp}
    \widehat{\theta}_\lambda - \theta_\star
    =
    V_\lambda^{-1} M_B \;-\; \lambda V_\lambda^{-1}\theta_\star,
\end{equation}
where $M_B=\sum_{t=1}^B\phi_t\varepsilon_t$ as in~\eqref{eq:V_def}.
Take the $V_\lambda$-norm:
\[
\|\widehat{\theta}_\lambda-\theta_\star\|_{V_\lambda}
\le
\|V_\lambda^{-1}M_B\|_{V_\lambda} + \lambda\|V_\lambda^{-1}\theta_\star\|_{V_\lambda}.
\]
Using $\|V_\lambda^{-1}M_B\|_{V_\lambda} = \|M_B\|_{V_\lambda^{-1}}$ and
$\lambda\|V_\lambda^{-1}\theta_\star\|_{V_\lambda}
= \lambda\sqrt{\theta_\star^\top V_\lambda^{-1}\theta_\star}
\le \lambda\sqrt{\theta_\star^\top (\lambda I_d)^{-1}\theta_\star}
= \sqrt{\lambda}\,\|\theta_\star\|_2
\le \sqrt{\lambda}\,S$,
we obtain
\begin{equation}
    \label{eq:tri}
    \|\widehat{\theta}_\lambda-\theta_\star\|_{V_\lambda}
\le
\|M_B\|_{V_\lambda^{-1}} + \sqrt{\lambda}S.
\end{equation}
By Lemma~\ref{lem:subg} and Lemma~\ref{lem:selfnorm}, with probability at least $1-\delta$,
\[
\|M_B\|_{V_\lambda^{-1}}
\le
\sigma\sqrt{2\log\!\left(\frac{\det(V_\lambda)^{1/2}}{\det(\lambda I_d)^{1/2}\,\delta}\right)}.
\]
Plugging into~\eqref{eq:tri} proves~\eqref{eq:self_norm}.

Finally, for any $x\in\mathcal{X}$,
\[
|\widehat{\tau}_\lambda(x)-\tau(x)|
=
|\phi(x)^\top(\widehat{\theta}_\lambda-\theta_\star)|
\le
\|\widehat{\theta}_\lambda-\theta_\star\|_{V_\lambda}\cdot \|\phi(x)\|_{V_\lambda^{-1}}
=
\|\widehat{\theta}_\lambda-\theta_\star\|_{V_\lambda}\cdot \sqrt{\phi(x)^\top V_\lambda^{-1}\phi(x)},
\]
which combined with~\eqref{eq:self_norm} yields~\eqref{eq:pointwise}.
\end{proof}

Theorem~\ref{thm:finite} quantifies a clean ``budget $\Rightarrow$ error'' tradeoff under \emph{adaptive} experimentation:
the estimation error is controlled by an \emph{information matrix} $V_\lambda$ built from the queried covariates.
The crucial difference from classical i.i.d. regression is that the $X_t$'s may be chosen adaptively based on past outcomes;
nevertheless, the deviation bound remains valid because the centered pseudo-outcome noise forms a martingale difference sequence.

\begin{corollary}[Finite-sample PEHE bound]
Under the conditions of Theorem~\ref{thm:finite}, with probability at least $1-\delta$,
\begin{equation}
    \mathcal{R}(\widehat{\tau}_\lambda)
    \;\le\;
    \beta_{B}(\delta)\cdot
    \sqrt{\mathbb{E}_{X\sim \mathbb{P}_X}\!\left[\phi(X)^\top V_\lambda^{-1}\phi(X)\right]}.
\end{equation}
In particular, if the queried design is \emph{well-conditioned} in the sense that
$V_0 \succeq \kappa B\, I_d$ for some $\kappa>0$, and $\Sigma_X\triangleq \mathbb{E}_{\mathbb{P}_X}[\phi(X)\phi(X)^\top]\preceq L^2 I_d$,
then (taking $\lambda=0$ for simplicity and assuming $V_0$ is invertible),
\begin{equation}
    \mathcal{R}(\widehat{\tau}_0)
    \;\le\;
    \beta_{B}(\delta)\cdot \sqrt{\frac{\mathrm{Tr}(\Sigma_X)}{\kappa B}}
    \;\lesssim\;
    \frac{L\,\sigma}{\sqrt{\kappa}}\sqrt{\frac{d\log(B/\delta)}{B}}.
\end{equation}
\end{corollary}

\begin{proof}
For any fixed $x$, by Cauchy--Schwarz in the $V_\lambda$-inner product,
\[
(\widehat{\tau}_\lambda(x)-\tau(x))^2
=
(\phi(x)^\top(\widehat{\theta}_\lambda-\theta_\star))^2
\le
\|\widehat{\theta}_\lambda-\theta_\star\|_{V_\lambda}^2\cdot \phi(x)^\top V_\lambda^{-1}\phi(x).
\]
Taking expectation over $X\sim \mathbb{P}_X$ gives
\[
\mathbb{E}_{X\sim \mathbb{P}_X}\!\left[(\widehat{\tau}_\lambda(X)-\tau(X))^2\right]
\le
\|\widehat{\theta}_\lambda-\theta_\star\|_{V_\lambda}^2\cdot
\mathbb{E}_{X\sim \mathbb{P}_X}\!\left[\phi(X)^\top V_\lambda^{-1}\phi(X)\right].
\]
Taking square roots and using the high-probability bound
$\|\widehat{\theta}_\lambda-\theta_\star\|_{V_\lambda}\le \beta_B(\delta)$ from Theorem~\ref{thm:finite}
yields~\eqref{eq:pehe_bound}.

For the ``well-conditioned'' simplification, $V_0\succeq \kappa B I_d$ implies $V_0^{-1}\preceq (\kappa B)^{-1}I_d$.
Hence
\[
\mathbb{E}_{X\sim \mathbb{P}_X}\!\left[\phi(X)^\top V_0^{-1}\phi(X)\right]
=
\mathrm{Tr}\!\left(V_0^{-1}\,\Sigma_X\right)
\le
\mathrm{Tr}\!\left((\kappa B)^{-1}I_d\cdot \Sigma_X\right)
=
\frac{\mathrm{Tr}(\Sigma_X)}{\kappa B}.
\]
If additionally $\Sigma_X\preceq L^2 I_d$, then $\mathrm{Tr}(\Sigma_X)\le dL^2$, which yields the final scaling in~\eqref{eq:rate_simple}
(up to the log factor hidden in $\beta_B(\delta)$).
\end{proof}

Corollary~\ref{cor:pehe} shows that the PEHE risk is governed by an \emph{integrated leverage} term
$\mathbb{E}_{\mathbb{P}_X}[\phi(X)^\top V_\lambda^{-1}\phi(X)]$.
Active sampling is useful precisely because it can shape $V_\lambda$:
selecting points that increase eigenvalues of $V_\lambda$ reduces this term.
This provides a direct theoretical rationale for using an ``uncertainty'' proxy (Algorithm~\ref{alg:active_fusion}, $v_u$):
in linear models, predictive variance is proportional to $\phi(u)^\top V_\lambda^{-1}\phi(u)$,
and ensemble disagreement is a practical surrogate for that quantity.

\subsection{Asymptotic normality under adaptive sampling}

Finite-sample concentration ensures ``small error with high probability.''
For statistical inference (e.g., confidence intervals), we further need a CLT.
The non-i.i.d. nature of active sampling prevents a direct appeal to the classical i.i.d. CLT,
but a martingale CLT applies.

\begin{assumption}{(Stabilizing design and moments).}
As $B\to\infty$, the normalized information matrix converges in probability:
\begin{equation}
    \label{eq:stabilizeV}
    \frac{1}{B}\sum_{t=1}^{B}\phi_t\phi_t^\top \;\xrightarrow[]{p}\; \Sigma_\pi,
    \qquad \Sigma_\pi \succ 0.
\end{equation}
Moreover, the conditional quadratic variation stabilizes:
\begin{equation}
    \label{eq:stabilizeVar}
    \frac{1}{B}\sum_{t=1}^{B}
    \mathbb{E}\!\left[\varepsilon_t^2\,\phi_t\phi_t^\top \,\middle|\,\mathcal{F}_{t-1}\right]
    \;\xrightarrow[]{p}\; \Omega_\pi,
\end{equation}
and a Lindeberg condition holds: for every $\eta>0$,
\begin{equation}
    \label{eq:lindeberg}
    \frac{1}{B}\sum_{t=1}^{B}
    \mathbb{E}\!\left[\varepsilon_t^2\|\phi_t\|_2^2\cdot \mathbf{1}\!\left\{|\varepsilon_t|\|\phi_t\|_2>\eta\sqrt{B}\right\}\,\middle|\,\mathcal{F}_{t-1}\right]
    \xrightarrow[]{p} 0.
\end{equation}
\end{assumption}

\begin{theorem}[Asymptotic normality (martingale CLT)]
Suppose Assumptions~\ref{assump:predictable}--\ref{assump:stabilize} hold and $\lambda=0$.
Then
\begin{equation}
    \sqrt{B}\big(\widehat{\theta}_0-\theta_\star\big)
    \;\xRightarrow[]{d}\;
    \mathcal{N}\!\left(0,\;\Sigma_\pi^{-1}\Omega_\pi\Sigma_\pi^{-1}\right).
\end{equation}
Consequently, for any fixed $x\in\mathcal{X}$,
\begin{equation}
    \sqrt{B}\big(\widehat{\tau}_0(x)-\tau(x)\big)
    \;\xRightarrow[]{d}\;
    \mathcal{N}\!\left(0,\;\phi(x)^\top\Sigma_\pi^{-1}\Omega_\pi\Sigma_\pi^{-1}\phi(x)\right).
\end{equation}
\end{theorem}

\begin{proof}
Recall $M_B=\sum_{t=1}^B \phi_t\varepsilon_t$ from~\eqref{eq:V_def} and $V_0=\sum_{t=1}^B\phi_t\phi_t^\top$.
From~\eqref{eq:theta_error_decomp} with $\lambda=0$ we have the exact representation
\begin{equation}
    \label{eq:ols_rep}
    \widehat{\theta}_0-\theta_\star = V_0^{-1}M_B.
\end{equation}

\textbf{Step 1: multivariate martingale CLT for $B^{-1/2}M_B$.}
We apply the Cram\'er--Wold device.
Fix any deterministic vector $a\in\mathbb{R}^d$ and define the scalar martingale difference array
\[
\xi_{t,B}(a)\triangleq a^\top\phi_t\,\varepsilon_t,
\qquad
S_B(a)\triangleq \sum_{t=1}^B \xi_{t,B}(a) = a^\top M_B.
\]
Since $\phi_t$ is $\mathcal{F}_{t-1}$-measurable and $\mathbb{E}[\varepsilon_t\mid \mathcal{F}_{t-1}]=0$,
we have $\mathbb{E}[\xi_{t,B}(a)\mid \mathcal{F}_{t-1}]=0$.
The predictable quadratic variation of $S_B(a)$ is
\[
\langle S(a)\rangle_B
=
\sum_{t=1}^B \mathbb{E}\!\left[\xi_{t,B}(a)^2\mid \mathcal{F}_{t-1}\right]
=
\sum_{t=1}^B \mathbb{E}\!\left[\varepsilon_t^2\,(a^\top\phi_t)^2\mid \mathcal{F}_{t-1}\right]
=
a^\top\left(\sum_{t=1}^B\mathbb{E}\!\left[\varepsilon_t^2\phi_t\phi_t^\top\mid \mathcal{F}_{t-1}\right]\right)a.
\]
By Assumption~\ref{assump:stabilize} (specifically~\eqref{eq:stabilizeVar}),
\[
\frac{1}{B}\langle S(a)\rangle_B \xrightarrow[]{p} a^\top \Omega_\pi a.
\]
Moreover, the Lindeberg condition~\eqref{eq:lindeberg} implies the (scalar) Lindeberg condition for $\xi_{t,B}(a)$:
for any $\eta>0$,
\[
\frac{1}{B}\sum_{t=1}^B
\mathbb{E}\!\left[\xi_{t,B}(a)^2\cdot \mathbf{1}\!\left\{|\xi_{t,B}(a)|>\eta\sqrt{B}\right\}\,\middle|\,\mathcal{F}_{t-1}\right]
\;\le\;
\|a\|_2^2\cdot
\frac{1}{B}\sum_{t=1}^B
\mathbb{E}\!\left[\varepsilon_t^2\|\phi_t\|_2^2\cdot \mathbf{1}\!\left\{|\varepsilon_t|\|\phi_t\|_2>\eta\sqrt{B}/\|a\|_2\right\}\,\middle|\,\mathcal{F}_{t-1}\right]
\xrightarrow[]{p}0.
\]
Therefore, by a martingale central limit theorem for triangular arrays (e.g., the Hall--Heyde martingale CLT),
\[
\frac{1}{\sqrt{B}}S_B(a) = a^\top\frac{1}{\sqrt{B}}M_B \;\xRightarrow[]{d}\; \mathcal{N}\!\left(0,\,a^\top\Omega_\pi a\right).
\]
Since this holds for every fixed $a$, the Cram\'er--Wold device yields the multivariate convergence
\begin{equation}
    \label{eq:MB_clt}
    \frac{1}{\sqrt{B}}M_B \;\xRightarrow[]{d}\; \mathcal{N}(0,\Omega_\pi).
\end{equation}

\textbf{Step 2: Slutsky with the stabilizing design.}
By Assumption~\ref{assump:stabilize} in~\eqref{eq:stabilizeV}, we have $B^{-1}V_0 \xrightarrow[]{p}\Sigma_\pi$ and $\Sigma_\pi\succ 0$.
By the continuous mapping theorem, $(B^{-1}V_0)^{-1}\xrightarrow[]{p}\Sigma_\pi^{-1}$.
Now multiply~\eqref{eq:ols_rep} by $\sqrt{B}$:
\[
\sqrt{B}(\widehat{\theta}_0-\theta_\star)
=
\left(\frac{1}{B}V_0\right)^{-1}\left(\frac{1}{\sqrt{B}}M_B\right).
\]
Combine the convergence in probability of $(B^{-1}V_0)^{-1}$ with the distributional convergence~\eqref{eq:MB_clt} and apply Slutsky's theorem to obtain
\[
\sqrt{B}(\widehat{\theta}_0-\theta_\star)
\;\xRightarrow[]{d}\;
\Sigma_\pi^{-1}\,Z,
\qquad Z\sim\mathcal{N}(0,\Omega_\pi),
\]
which implies
$\Sigma_\pi^{-1}Z\sim \mathcal{N}(0,\Sigma_\pi^{-1}\Omega_\pi\Sigma_\pi^{-1})$.

\textbf{Step 3: pointwise CATE normality.}
For any fixed $x$, $\widehat{\tau}_0(x)-\tau(x) = \phi(x)^\top(\widehat{\theta}_0-\theta_\star)$.
Apply the continuous mapping theorem to the linear functional $z\mapsto \phi(x)^\top z$ to conclude the second statement.
\end{proof}

Theorem~\ref{thm:clt} says: \emph{even though the RCT data are collected adaptively and are non-i.i.d.,}
the final estimator still admits classical $\sqrt{B}$-asymptotics.
A common misconception is ``CLTs require i.i.d.''; here i.i.d.\ is sufficient but not necessary.
The deeper requirement is the martingale structure created by randomization (mean-zero increments) and a non-degenerate stabilizing design.

\subsection{Minimax lower bound and (near) optimality}

Upper bounds alone do not tell us whether a method is ``the best possible.''
We therefore complement them with a minimax lower bound.
The message is intentionally sharp:
\emph{even with perfect adaptivity and access to unlimited observational logs and unlabeled pools, one cannot beat the $\sqrt{d/B}$ scaling in general.}

\paragraph{A hard instance family.}
We consider a discrete covariate space consisting of $d$ ``types'' $\{x^{(1)},\dots,x^{(d)}\}$ with
\begin{equation}
    \label{eq:PX_uniform}
    \mathbb{P}_X(x^{(j)}) = \frac{1}{d}, \qquad j=1,\dots,d,
    \qquad\text{and}\qquad
    \phi(x^{(j)}) = e_j \in \mathbb{R}^d.
\end{equation}
Define the parameter set $\Theta_\Delta \triangleq \{-\Delta,+\Delta\}^d$ for some $\Delta\in(0,1)$.
For each $\theta\in\Theta_\Delta$, define potential outcomes (conditionally independent across units and over time) by
\begin{equation}
    \label{eq:bernoulli_hard}
    Y(1)\mid X=x^{(j)} \sim \mathrm{Bern}\!\left(\frac{1}{2}+\frac{\theta_j}{2}\right),
    \qquad
    Y(0)\mid X=x^{(j)} \sim \mathrm{Bern}\!\left(\frac{1}{2}-\frac{\theta_j}{2}\right).
\end{equation}
Then $\tau(x^{(j)})=\theta_j$ and $\|\theta\|_2 = \sqrt{d}\Delta$.
(If the function class enforces $\|\theta\|_2\le S$, we choose $\Delta\le S/\sqrt{d}$ so that $\Theta_\Delta\subseteq\{\|\theta\|_2\le S\}$.)
This family lies entirely within the binary-outcome potential-outcomes model and is compatible with bounded randomization.

\paragraph{A basic two-point inequality.}
We will use the following standard reduction from estimation to testing.

\begin{lemma}[Two-point lower bound via total variation]
\label{lem:two_point}
Let $P$ and $Q$ be two probability measures and let $a\neq b$ be two real numbers.
For any (possibly randomized) estimator $\widehat{u}$ based on an observation distributed as either $P$ or $Q$,
\begin{equation}
    \label{eq:two_point_tv}
    \max\left\{\mathbb{E}_{P}\!\left[(\widehat{u}-a)^2\right],\ \mathbb{E}_{Q}\!\left[(\widehat{u}-b)^2\right]\right\}
    \;\ge\;
    \frac{(a-b)^2}{8}\Big(1-\mathrm{TV}(P,Q)\Big),
\end{equation}
where $\mathrm{TV}(P,Q)\triangleq \sup_{A} |P(A)-Q(A)|$.
\end{lemma}

\begin{proof}
Let $\psi$ be the test $\psi=1$ if $|\widehat u-a|\le |\widehat u-b|$ and $\psi=0$ otherwise.
If $\psi=1$, then $|\widehat u-b|\ge |a-b|/2$; if $\psi=0$, then $|\widehat u-a|\ge |a-b|/2$.
Hence,
\[
(\widehat u-a)^2 \ge \frac{(a-b)^2}{4}\mathbf{1}\{\psi=0\},\qquad
(\widehat u-b)^2 \ge \frac{(a-b)^2}{4}\mathbf{1}\{\psi=1\}.
\]
Therefore,
\[
\mathbb{E}_{P}[(\widehat u-a)^2] + \mathbb{E}_{Q}[(\widehat u-b)^2]
\;\ge\;
\frac{(a-b)^2}{4}\Big(\mathbb{P}_{P}(\psi=0) + \mathbb{P}_{Q}(\psi=1)\Big).
\]
By the Neyman--Pearson characterization of total variation,
for any test $\psi$,
$\mathbb{P}_{P}(\psi=0)+\mathbb{P}_{Q}(\psi=1)\ge 1-\mathrm{TV}(P,Q)$.
Thus
\[
\mathbb{E}_{P}[(\widehat u-a)^2] + \mathbb{E}_{Q}[(\widehat u-b)^2]
\;\ge\;
\frac{(a-b)^2}{4}\Big(1-\mathrm{TV}(P,Q)\Big).
\]
Taking the maximum of the two terms on the left-hand side yields~\eqref{eq:two_point_tv}.
\end{proof}

\paragraph{A chain rule for KL under adaptive designs.}
Let $Z_{1:B}$ denote the entire observed data stream generated by a fixed adaptive policy $\pi$:
$Z_t$ includes $(X_t,p_t,T_t,Y_t)$ (and can include any extra bookkeeping variables).
Write $P_{\theta}^{\pi}$ for the induced law of $Z_{1:B}$ under parameter $\theta$ and policy $\pi$.

\begin{lemma}[Sequential chain rule for KL under a fixed policy]
\label{lem:kl_chain}
For any $\theta,\theta'\in\Theta_\Delta$ and any fixed policy $\pi$,
\begin{equation}
    \label{eq:kl_chain}
    \mathrm{KL}\!\left(P_{\theta}^{\pi}\,\|\,P_{\theta'}^{\pi}\right)
    =
    \sum_{t=1}^{B}
    \mathbb{E}_{\theta}^{\pi}\!\left[
        \mathrm{KL}\!\left(
            P_{\theta}^{\pi}(Z_t\mid \mathcal{F}_{t-1})
            \,\big\|\,
            P_{\theta'}^{\pi}(Z_t\mid \mathcal{F}_{t-1})
        \right)
    \right].
\end{equation}
\end{lemma}

\begin{proof}
This is the standard KL chain rule for a joint law written as a product of conditionals.
Because the policy $\pi$ is fixed, both $P_{\theta}^{\pi}$ and $P_{\theta'}^{\pi}$ admit the same filtration $\{\mathcal{F}_t\}$,
and their joint densities (w.r.t.\ a suitable product dominating measure) factorize as
$p_\theta(z_{1:B})=\prod_{t=1}^B p_\theta(z_t\mid z_{1:t-1})$ and similarly for $\theta'$.
Then
\begin{align*}
\mathrm{KL}(P_\theta^\pi\|P_{\theta'}^\pi)
&=
\mathbb{E}_{\theta}^{\pi}\!\left[\log\frac{p_\theta(Z_{1:B})}{p_{\theta'}(Z_{1:B})}\right]
=
\mathbb{E}_{\theta}^{\pi}\!\left[\sum_{t=1}^B \log\frac{p_\theta(Z_t\mid Z_{1:t-1})}{p_{\theta'}(Z_t\mid Z_{1:t-1})}\right] \\
&=
\sum_{t=1}^B
\mathbb{E}_{\theta}^{\pi}\!\left[
\mathbb{E}_{\theta}^{\pi}\!\left[\log\frac{p_\theta(Z_t\mid \mathcal{F}_{t-1})}{p_{\theta'}(Z_t\mid \mathcal{F}_{t-1})}\,\middle|\,\mathcal{F}_{t-1}\right]
\right] \\
&=
\sum_{t=1}^{B}
\mathbb{E}_{\theta}^{\pi}\!\left[
        \mathrm{KL}\!\left(
            P_{\theta}^{\pi}(Z_t\mid \mathcal{F}_{t-1})
            \,\big\|\,
            P_{\theta'}^{\pi}(Z_t\mid \mathcal{F}_{t-1})
        \right)
\right],
\end{align*}
which is~\eqref{eq:kl_chain}.
\end{proof}

\begin{theorem}[Minimax lower bound for active RCT (rate optimality)]
Consider the hard family~\eqref{eq:PX_uniform}--\eqref{eq:bernoulli_hard} with $\Delta\le 1/2$ and $\sqrt{d}\Delta\le S$.
There exist universal constants $c_0,c_1>0$ such that if $B\ge c_0 d$, then for any adaptive querying policy $\pi$
(possibly using $\mathcal{D}_{\mathrm{obs}}$ and $\mathcal{D}_{\mathrm{pool}}$) and any estimator $\widehat{\tau}$ based on $B$ RCT samples,
\begin{equation}
    \label{eq:minimax_clean}
    \inf_{\pi}\ \inf_{\widehat{\tau}}\ \sup_{\theta\in\Theta_\Delta}
    \mathbb{E}_{\theta}^{\pi}\!\left[\mathcal{R}(\widehat{\tau})\right]
    \;\ge\;
    c_1\sqrt{\frac{d}{B}}.
\end{equation}
In particular, the $\sqrt{d/B}$ scaling in Corollary~\ref{cor:pehe} is information-theoretically unimprovable in general.
\end{theorem}

\begin{proof}
We prove a lower bound on the \emph{expected squared} PEHE risk and then take square roots.

\textbf{Step 1: reduce minimax risk to a Bayes risk over a hypercube prior.}
Let $\Pi$ be the uniform prior over $\Theta_\Delta=\{-\Delta,+\Delta\}^d$.
For any policy $\pi$ and estimator $\widehat\tau$, by Jensen and the fact that $\sup_\theta \ge \mathbb{E}_{\theta\sim\Pi}$,
\[
\sup_{\theta\in\Theta_\Delta}\mathbb{E}_{\theta}^{\pi}\!\left[\mathcal{R}(\widehat\tau)^2\right]
\;\ge\;
\mathbb{E}_{\theta\sim\Pi}\mathbb{E}_{\theta}^{\pi}\!\left[\mathcal{R}(\widehat\tau)^2\right].
\]
Thus it suffices to lower bound the Bayes risk under $\Pi$ uniformly over $\pi$ and $\widehat\tau$.

Under~\eqref{eq:PX_uniform} and $\phi(x^{(j)})=e_j$, the PEHE squared is
\[
\mathcal{R}(\widehat\tau)^2
=
\mathbb{E}_{X\sim \mathbb{P}_X}\!\left[(\widehat\tau(X)-\tau(X))^2\right]
=
\frac{1}{d}\sum_{j=1}^d \left(\widehat\tau(x^{(j)})-\theta_j\right)^2.
\]
Define $\widehat\theta_j\triangleq \widehat\tau(x^{(j)})$.
Then
\begin{equation}
    \label{eq:pehe_theta}
    \mathcal{R}(\widehat\tau)^2 = \frac{1}{d}\sum_{j=1}^d (\widehat\theta_j-\theta_j)^2.
\end{equation}

\textbf{Step 2: apply the two-point bound coordinate-wise (Assouad-style).}
For each $j$, define the ``flipped'' parameter $\theta^{(j)}$ by
$\theta^{(j)}_j=-\theta_j$ and $\theta^{(j)}_\ell=\theta_\ell$ for $\ell\neq j$.
Let $P_{\theta}^{\pi}$ and $P_{\theta^{(j)}}^{\pi}$ be the laws of the entire adaptive data stream under these two parameters.

Apply Lemma~\ref{lem:two_point} with $a=\Delta$ and $b=-\Delta$ to the estimation of coordinate $j$,
and then average over the uniform prior on $\theta$.
A standard symmetrization argument yields
\begin{equation}
    \label{eq:assouad_step}
    \mathbb{E}_{\theta\sim\Pi}\mathbb{E}_{\theta}^{\pi}\!\left[(\widehat\theta_j-\theta_j)^2\right]
    \;\ge\;
    \frac{\Delta^2}{8}\left(1 - \mathbb{E}_{\theta\sim\Pi}\left[\mathrm{TV}\!\left(P_{\theta}^{\pi},P_{\theta^{(j)}}^{\pi}\right)\right]\right).
\end{equation}
Summing~\eqref{eq:assouad_step} over $j=1,\dots,d$ and dividing by $d$,
and using~\eqref{eq:pehe_theta}, we get
\begin{equation}
    \label{eq:bayes_pehe_lb}
    \mathbb{E}_{\theta\sim\Pi}\mathbb{E}_{\theta}^{\pi}\!\left[\mathcal{R}(\widehat\tau)^2\right]
    \;\ge\;
    \frac{\Delta^2}{8}\left(
        1 - \frac{1}{d}\sum_{j=1}^d
        \mathbb{E}_{\theta\sim\Pi}\left[\mathrm{TV}\!\left(P_{\theta}^{\pi},P_{\theta^{(j)}}^{\pi}\right)\right]
    \right).
\end{equation}

\textbf{Step 3: bound total variation by KL and control KL under adaptivity.}
By Pinsker's inequality, $\mathrm{TV}(P,Q)\le \sqrt{\mathrm{KL}(P\|Q)/2}$.
Thus
\[
\frac{1}{d}\sum_{j=1}^d \mathbb{E}_{\theta\sim\Pi}\left[\mathrm{TV}\!\left(P_{\theta}^{\pi},P_{\theta^{(j)}}^{\pi}\right)\right]
\le
\frac{1}{d}\sum_{j=1}^d
\mathbb{E}_{\theta\sim\Pi}\left[\sqrt{\frac{1}{2}\mathrm{KL}\!\left(P_{\theta}^{\pi}\,\|\,P_{\theta^{(j)}}^{\pi}\right)}\right].
\]
Since $x\mapsto \sqrt{x}$ is concave on $\mathbb{R}_+$, Jensen gives
\begin{equation}
    \label{eq:jensen_sqrt}
    \frac{1}{d}\sum_{j=1}^d
\mathbb{E}_{\theta\sim\Pi}\left[\sqrt{\frac{1}{2}\mathrm{KL}\!\left(P_{\theta}^{\pi}\,\|\,P_{\theta^{(j)}}^{\pi}\right)}\right]
\le
\sqrt{
\frac{1}{2d}\sum_{j=1}^d
\mathbb{E}_{\theta\sim\Pi}\left[\mathrm{KL}\!\left(P_{\theta}^{\pi}\,\|\,P_{\theta^{(j)}}^{\pi}\right)\right]
}.
\end{equation}

It remains to bound the average KL.
Fix $j$ and apply Lemma~\ref{lem:kl_chain}.
Because $\theta$ and $\theta^{(j)}$ differ only at coordinate $j$, the conditional distributions of $Y_t$ coincide whenever $X_t\neq x^{(j)}$.
Hence the conditional KL contribution at time $t$ is zero if $X_t\neq x^{(j)}$.
If $X_t=x^{(j)}$, then the conditional distribution of $(T_t,Y_t)$ differs only through the Bernoulli parameters
$(1/2\pm \Delta/2)$ in~\eqref{eq:bernoulli_hard}.
A direct computation (or the standard Bernoulli KL bound) shows that for $\Delta\le 1/2$ there exists a universal constant $C_{\mathrm{KL}}>0$ such that
\begin{equation}
    \label{eq:ber_kl_bound}
    \mathrm{KL}\!\left(
    \mathrm{Bern}\!\left(\frac{1}{2}+\frac{\Delta}{2}\right)
    \,\bigg\|\,
    \mathrm{Bern}\!\left(\frac{1}{2}-\frac{\Delta}{2}\right)
    \right)
    \;\le\; C_{\mathrm{KL}}\,\Delta^2,
    \qquad
    \mathrm{KL}\!\left(
    \mathrm{Bern}\!\left(\frac{1}{2}-\frac{\Delta}{2}\right)
    \,\bigg\|\,
    \mathrm{Bern}\!\left(\frac{1}{2}+\frac{\Delta}{2}\right)
    \right)
    \;\le\; C_{\mathrm{KL}}\,\Delta^2.
\end{equation}
(For instance, one may take $C_{\mathrm{KL}}=16/3$ since the Bernoulli parameters lie in $[1/4,3/4]$ and
$\mathrm{KL}(\mathrm{Bern}(p)\|\mathrm{Bern}(q))\le (p-q)^2/(q(1-q))$.)

Therefore, the one-step KL between the conditional laws of $(T_t,Y_t)$ under $\theta$ and $\theta^{(j)}$ is at most $C_{\mathrm{KL}}\Delta^2$ whenever $X_t=x^{(j)}$.
Let $N_j\triangleq \sum_{t=1}^B \mathbf{1}\{X_t=x^{(j)}\}$ be the (random) number of times type $j$ is queried.
Combining with the chain rule~\eqref{eq:kl_chain} yields
\begin{equation}
    \label{eq:KL_Nj}
    \mathrm{KL}\!\left(P_{\theta}^{\pi}\,\|\,P_{\theta^{(j)}}^{\pi}\right)
    \;\le\;
    C_{\mathrm{KL}}\,\Delta^2\;\mathbb{E}_{\theta}^{\pi}[N_j].
\end{equation}
Now average~\eqref{eq:KL_Nj} over $\theta\sim\Pi$ and sum over $j$:
since $\sum_{j=1}^d N_j = B$ deterministically (exactly one $X_t$ is queried per round),
we have for every $\theta$,
$\sum_{j=1}^d \mathbb{E}_{\theta}^{\pi}[N_j]=B$,
and thus
\begin{equation}
    \label{eq:avgKL}
    \frac{1}{d}\sum_{j=1}^d
    \mathbb{E}_{\theta\sim\Pi}\left[\mathrm{KL}\!\left(P_{\theta}^{\pi}\,\|\,P_{\theta^{(j)}}^{\pi}\right)\right]
    \;\le\;
    \frac{C_{\mathrm{KL}}\Delta^2}{d}\sum_{j=1}^d \mathbb{E}_{\theta\sim\Pi}\mathbb{E}_{\theta}^{\pi}[N_j]
    =
    C_{\mathrm{KL}}\Delta^2\cdot \frac{B}{d}.
\end{equation}

Plugging~\eqref{eq:avgKL} into~\eqref{eq:jensen_sqrt} gives
\begin{equation}
    \label{eq:avgTV}
    \frac{1}{d}\sum_{j=1}^d \mathbb{E}_{\theta\sim\Pi}\left[\mathrm{TV}\!\left(P_{\theta}^{\pi},P_{\theta^{(j)}}^{\pi}\right)\right]
    \;\le\;
    \sqrt{\frac{C_{\mathrm{KL}}\Delta^2 B}{2d}}.
\end{equation}

\textbf{Step 4: choose $\Delta$ and conclude.}
Choose
\[
\Delta^2 \;=\; \min\left\{\frac{1}{4},\ \frac{d}{16C_{\mathrm{KL}}B},\ \frac{S^2}{d}\right\}.
\]
Then $\Delta\le 1/2$ and $\sqrt{d}\Delta\le S$.
Moreover, if $B\ge c_0 d$ for a sufficiently large universal $c_0$,
then $\Delta^2 = \Theta(d/B)$ and~\eqref{eq:avgTV} implies the RHS is at most, say, $1/4$.
Plugging this into~\eqref{eq:bayes_pehe_lb} yields
\[
\mathbb{E}_{\theta\sim\Pi}\mathbb{E}_{\theta}^{\pi}\!\left[\mathcal{R}(\widehat\tau)^2\right]
\;\ge\;
\frac{\Delta^2}{8}\left(1-\frac{1}{4}\right)
\;\ge\;
c\cdot \frac{d}{B}
\]
for a universal constant $c>0$.
Therefore,
\[
\mathbb{E}_{\theta}^{\pi}\!\left[\mathcal{R}(\widehat\tau)\right]
\;\ge\;
\sqrt{\mathbb{E}_{\theta}^{\pi}\!\left[\mathcal{R}(\widehat\tau)^2\right]}
\;\ge\;
c_1\sqrt{\frac{d}{B}},
\]
where we used Jensen ($\mathbb{E}[\sqrt{Z}]\le \sqrt{\mathbb{E}[Z]}$) in reverse by lower-bounding $\mathbb{E}[\mathcal{R}^2]$ and then taking square roots.
Since the above bound holds for the Bayes risk under $\Pi$, it also lower bounds the minimax risk,
and since $\pi$ and $\widehat\tau$ were arbitrary, the infima over them preserve the bound.
This proves~\eqref{eq:minimax_clean}.
\end{proof}

Theorem~\ref{thm:minimax} formalizes a simple but important reality:
if there are $d$ independent degrees of freedom in $\tau(\cdot)$ (e.g., $d$ user segments with genuinely different causal responses),
then with budget $B$ one cannot estimate all of them faster than order $\sqrt{d/B}$ in PEHE.
Active learning can \emph{reallocate} samples to reduce constants and improve conditioning,
but it cannot create information out of thin air.

\begin{corollary}[Near-minimax optimality of the orthogonalized estimator]
Under Assumptions~\ref{assump:predictable}--\ref{assump:linear} and a well-conditioned design $V_0\succeq \kappa B I_d$,
the upper bound in~\eqref{eq:rate_simple} matches the minimax lower bound~\eqref{eq:minimax_clean}
up to logarithmic factors and the conditioning constant $\kappa$.
In this sense, the estimator in Algorithm~\ref{alg:theory_est} is minimax-rate optimal (up to logs) for the linear class.
\end{corollary}

\begin{proof}
Combine Corollary~\ref{cor:pehe} (upper bound) and Theorem~\ref{thm:minimax} (lower bound).
Both scale as $\sqrt{d/B}$, and the only discrepancy is the logarithmic term in $\beta_B(\delta)$ (and constants depending on $\kappa$ and boundedness).
\end{proof}

\subsection{How the main results connect}

We close the theory section by clarifying the logical flow:
\begin{itemize}
    \item Lemma~\ref{lem:unbiased_pseudo} is the \textbf{foundation}: it isolates the causal ``truth'' created by randomization, even under adaptive selection.
    \item Lemmas~\ref{lem:subg}--\ref{lem:selfnorm} provide the \textbf{technical engine}: self-normalized martingale concentration that remains valid for non-i.i.d.\ adaptive designs.
    \item Theorem~\ref{thm:finite} is the \textbf{strongest finite-sample guarantee}: a deviation bound controlled by the information matrix $V_\lambda$.
    \item Corollary~\ref{cor:pehe} translates this into the paper's \textbf{target metric} (PEHE risk), revealing the central role of integrated leverage.
    \item Theorem~\ref{thm:clt} upgrades consistency to \textbf{inference}: asymptotic normality under a stabilizing adaptive design.
    \item Theorem~\ref{thm:minimax} provides the \textbf{impossibility result}: no method can beat the $\sqrt{d/B}$ scaling in general.
    \item Corollary~\ref{cor:optimality} combines the upper and lower bounds to conclude \textbf{(near) minimax optimality}.
\end{itemize}

\paragraph{(Optional) Choosing the randomization probability.}
If operational constraints permit, one may ask how to choose $p$ to minimize the pseudo-outcome variance.
The exact variance-optimal choice depends on second moments of potential outcomes:

\begin{proposition}[Variance-optimal randomization for the pseudo-outcome]
\label{prop:optimal_p}
Fix $x$ and denote $A(x)\triangleq \mathbb{E}[Y(1)^2\mid X=x]$ and $B(x)\triangleq \mathbb{E}[Y(0)^2\mid X=x]$.
Under Assumption~\ref{assump:rct_bound}, the conditional variance $\mathrm{Var}(\widetilde{Y}\mid X=x,p)$ is minimized at
\begin{equation}
    \label{eq:p_opt}
    p^\star(x) \;=\; \frac{\sqrt{A(x)}}{\sqrt{A(x)}+\sqrt{B(x)}}.
\end{equation}
In particular, if $A(x)=B(x)$ (e.g., comparable outcome scales across arms), then $p^\star(x)=1/2$.
\end{proposition}

\begin{proof}
Condition on $X=x$ and $p\in(0,1)$.
Using the definition of pseudo outcome and the randomization $T\sim\mathrm{Bern}(p)$,
\[
\widetilde{Y}
=
\begin{cases}
Y(1)/p, & T=1,\\
-\,Y(0)/(1-p), & T=0.
\end{cases}
\]
Hence
\[
\mathbb{E}[\widetilde{Y}^2\mid X=x,p]
=
p\cdot \mathbb{E}\!\left[\frac{Y(1)^2}{p^2}\,\middle|\,X=x\right]
+
(1-p)\cdot \mathbb{E}\!\left[\frac{Y(0)^2}{(1-p)^2}\,\middle|\,X=x\right]
=
\frac{A(x)}{p}+\frac{B(x)}{1-p}.
\]
Since $\mathbb{E}[\widetilde{Y}\mid X=x,p]=\tau(x)$ does not depend on $p$ (Lemma~\ref{lem:unbiased_pseudo}),
minimizing $\mathrm{Var}(\widetilde{Y}\mid X=x,p)$ over $p$ is equivalent to minimizing
$f(p)=A(x)/p + B(x)/(1-p)$ over $p\in(0,1)$.
This is a strictly convex function on $(0,1)$ with derivative
$f'(p)=-A(x)/p^2 + B(x)/(1-p)^2$.
Setting $f'(p)=0$ yields $(1-p)/p = \sqrt{B(x)/A(x)}$, which rearranges to~\eqref{eq:p_opt}.
The special case $A(x)=B(x)$ gives $p^\star(x)=1/2$.
\end{proof}

Proposition~\ref{prop:optimal_p} justifies $p\approx 1/2$ as a robust default when outcome scales are comparable across arms,
and motivates clipping $[f_{\min},f_{\max}]$ as a principled safeguard when $p$ must vary by covariates.

\section{Detailed Design of the Acquisition Strategy}
\label{app:acquisition_details}

In this section, we provide a rigorous formulation of our proposed active learning framework, specifically designed to enhance the Disentangled Representations for
CounterFactual Regression (DRCFR) model. We adopt DRCFR as the foundational backbone not merely for its theoretical properties, but due to its proven efficacy as our \textit{deployed production engine}, where it demonstrates superior stability in uplift estimation across diverse industrial scenarios.
Building on this robust baseline, our methodology integrates two core components: a multi-faceted acquisition function for strategic sample selection and a reverse alignment weighting mechanism for model optimization.

\subsection{Multi-Faceted Acquisition Function}

To select the most informative samples from the unlabeled candidate pool $\mathcal{D}_{\mathrm{pool}}$, we design a composite acquisition function that addresses three critical challenges in causal inference: \emph{epistemic uncertainty}, \emph{domain discrepancy}, and \emph{overlap deficit}. For any candidate sample $u \in \mathcal{D}_{\mathrm{pool}}$, the total acquisition score $S(u)$ is defined as:

\begin{equation}
    S(u) = \alpha \cdot \eta(v_u) + \beta \cdot \eta(d_u) + \gamma \cdot \eta(o_u),
\end{equation}

where $\alpha, \beta, \gamma$ are hyperparameters balancing the trade-off between uncertainty reduction, domain discrepancy, and overlap deficit. We detail each component below.

\paragraph{1. Epistemic Uncertainty Score ($v_u$).} 
To measure the model's lack of knowledge regarding the Conditional Average Treatment Effect (CATE), we employ Monte Carlo (MC) Dropout as a Bayesian approximation. Unlike standard output variance, we focus specifically on the variance of the predicted \textit{uplift}. 
Let $f_{\theta}$ denote the inference network with dropout enabled. We perform $E$ stochastic forward passes for each sample $u$, obtaining a set of CATE predictions $\{\hat{\tau}_j(u)\}_{j=1}^E$. The uncertainty score is quantified as the variance of these predictions:

\begin{equation}
    v_u = \mathrm{Var}\left( \{ \hat{\tau}_j(u) \}_{j=1}^E \right) = \frac{1}{E} \sum_{j=1}^{E} \left( \hat{\tau}_j(u) - \bar{\tau}(u) \right)^2,
\end{equation}

where $\bar{\tau}(u)$ is the empirical mean. A higher variance implies high epistemic uncertainty, suggesting that obtaining the ground truth label for $u$ would yield significant information gain.

\paragraph{2. Domain Discrepancy Score ($d_u$).}
To ensure the selected samples effectively cover the feature space of the target pool (exploration) and mitigate domain shift, we introduce a Domain Classifier. This module, parameterized by $\xi$, discriminates between the \emph{current} training distribution ($\mathcal{D}_{\mathrm{current}} = \mathcal{D}_{\mathrm{obs}} \cup \mathcal{D}_{\mathrm{rct}}$) and the unselected pool ($\mathcal{D}_{\mathrm{target}} = \mathcal{D}_{\mathrm{pool}}$).

Crucially, the classifier operates on the disentangled representation $\phi(u) \in \mathbb{R}^{d_{rep}}$ extracted from the DRCFR backbone, rather than raw features. Structurally, $g_\xi$ is designed as a two-layer Multilayer Perceptron (MLP). It consists of a linear transformation to a hidden dimension $d_{hidden}$, followed by an Exponential Linear Unit (ELU) activation and a Dropout layer. The final layer maps the hidden features to a scalar logit. Formally, the forward pass for a sample $u$ is defined as:

\begin{equation}
    \mathbf{h} = \mathrm{Dropout}(\mathrm{ELU}(\mathbf{W}_1 \phi(u) + \mathbf{b}_1)).
\end{equation}
\begin{equation}
    g_\xi(\phi(u)) = \mathbf{W}_2 \mathbf{h} + b_2.
\end{equation}

where $\mathbf{W}_1$ and $\mathbf{W}_2$ are learnable weights. The discrepancy score is then defined as the predicted probability of belonging to the target pool (label 1):

\begin{equation}
    d_u = \mathbb{P}(\mathrm{domain}=1 \mid \phi(u); \xi) = \sigma(g_\xi(\phi(u)))
\end{equation}

where $\sigma(\cdot)$ is the sigmoid function. A higher $d_u$ indicates that $u$ lies in a region under-represented by the current training set, necessitating active sampling for distribution alignment.

\paragraph{3. Overlap Deficit Score ($o_u$).}
To explicitly target samples heavily affected by selection bias in the observational data, we define a metric based on propensity extremity. We utilize a propensity head $\hat{e}_{\mathrm{obs}}(\phi(u))$, pre-trained strictly on $\mathcal{D}_{\mathrm{obs}}$, to capture the historical assignment mechanism. The score targets samples where the observational policy is deterministic (i.e., propensity close to 0 or 1):

\begin{equation}
    o_u = \left| \hat{e}_{\mathrm{obs}}(\phi(u)) - 0.5 \right| \times 2.
\end{equation}

Samples with $o_u \approx 1$ represent individuals who, historically, almost exclusively received one specific treatment (either control or treated). Querying these samples in an RCT context (where $e(\phi(u))=0.5$) provides critical counterfactual evidence to correct the learned bias.

\subsection{Optimization via Counterfactual Alignment and Sample Reweighting}

Once a batch of samples is selected and labeled (forming $\mathcal{D}_{\mathrm{rct}}$), we combine them with $\mathcal{D}_{\mathrm{obs}}$ to update the DRCFR model. Merely pooling the data is suboptimal; to maximize the utility of the RCT samples, we implement a \textbf{Counterfactual Alignment} strategy. This approach operationalizes the active sampling by reweighting instances that contradict observational bias, thereby prioritizing the acquisition of "missing" counterfactual information.

We define the \textit{Propensity-Assignment Gap}, $\Delta(u, t)$, as the absolute divergence between the realized randomized treatment assignment $t_{\mathrm{rct}}$ and the historical observational propensity $\hat{e}_{\mathrm{obs}}(\phi(u))$:

\begin{equation}
    \Delta(u, t_{\mathrm{rct}}) = |t_{\mathrm{rct}} - \hat{e}_{\mathrm{obs}}(\phi(u))|.
\end{equation}

A large gap (e.g., $\Delta > 0.5$) signifies that the RCT assignment acts as a counter-balance to the observational tendency. For instance, assigning control ($t=0$) to a patient who historically had a high probability of being treated ($\hat{e}_{\mathrm{obs}} \approx 1$) exposes outcomes in a region of the joint distribution that was previously a "blind spot."

To focus model training on these information-rich samples and debias the uplift estimator under a finite RCT budget, we apply the following weights $w_i$ to the loss function:

\begin{equation}
    w_i = 
    \begin{cases} 
    1.0, & \text{if } \Delta(u_i, t_i) > 0.5 \quad (\text{Counterfactual Alignment / Gold Sample}). \\
    0.2, & \text{otherwise} \quad (\text{Redundant / Silver Sample}).
    \end{cases}
\end{equation}

\textbf{Rationale:} Gold Samples (where the gap is large) directly fill the counterfactual void in $\mathcal{D}_{\mathrm{obs}}$, providing high information gain. Conversely, Silver Samples align with historical bias and merely reinforce existing knowledge; thus, they are down-weighted to prevent redundancy and ensure gradient updates are dominated by the counterfactual alignment signal.

Finally, the total training objective for the DRCFR backbone incorporates these alignment weights:

\begin{equation}
    \mathcal{L}_{DRCFR} = \sum_{i \in \mathcal{D}_{\mathrm{current}}} w_i \cdot \ell_{\mathrm{pred}}(y_i, \hat{y}_i) + \lambda_{\mathrm{disc}} \cdot \mathrm{MMD}(\phi(u)_{|t=0}, \phi(u)_{|t=1}) + \lambda_{\pi} \cdot \mathcal{L}_{\mathrm{prop}}.
\end{equation}

where the first term is the weighted factual prediction loss, the second term enforces representation balance via Maximum Mean Discrepancy (MMD), and the third term is the propensity regularization. In our experiments, we set the regularization coefficients $\lambda_{\mathrm{disc}} = \lambda_{\pi} = 1.0$.

\section{Experimental Setup}
\subsection{Detailed Description of the Dataset} \label{sec:data_details}

To comprehensively assess the active learning framework under extreme distribution shifts, realistic deployment scenarios, and temporal variations, we organize the data into distinct components following a strict \textbf{Out-of-Time (OOT) evaluation protocol} spanning from Nov.~25 to Jan.~07:

\begin{itemize}
    \item \textbf{Biased Observational Set ($\mathcal{D}_{\mathrm{obs}}^{\mathrm{bias}}$):} 
    Collected from the training phase (\textbf{Nov.~25 -- Dec.~25}), we curate a subset of approximately $4.5$ million historical OBS logs. Generated by deterministic production policies, this dataset exhibits severe selection bias and violations of the \textit{overlap assumption}. We utilize this dataset for: (i) comparative analysis between active sampling strategies and random baselines, and (ii) ablation studies on our proposed acquisition functions to evaluate robustness against sample imbalance.
    
    \item \textbf{Full-Scale Observational Set ($\mathcal{D}_{\mathrm{obs}}^{\mathrm{full}}$):} 
    To validate the framework's efficacy in a real business environment, we employ a larger dataset $\mathcal{D}_{\mathrm{obs}}^{\mathrm{full}}$ of approximately $6.5$ million logs from the same training period (\textbf{Nov.~25 -- Dec.~25}). Compared to $\mathcal{D}_{\mathrm{obs}}^{\mathrm{bias}}$, it shows much lower distributional shift and sample imbalance. Experiments on this set focus on evaluating performance sensitivity across varying ratios of OBS data.

    \item \textbf{RCT Candidate Pool ($\mathcal{D}_{\mathrm{pool}}$):} 
    Coinciding with the training phase (\textbf{Nov.~25 -- Dec.~25}), this dataset comprises approximately $1.2$ million samples representing the candidate pool for RCTs. Under our experimental setting, only covariate information $X$ is available beforehand, while treatment $T$ and outcome $Y$ remain \textit{masked} until explicitly queried. The active learner selectively queries labels subject to budgets $B \in \{10\text{k}, 50\text{k}, \dots, 500\text{k}\}$ to rectify covariate distribution shifts, yielding the labeled dataset $\mathcal{D}_{\mathrm{rct}}$.

    \item \textbf{RCT Test Set ($\mathcal{D}_{\mathrm{rct}}^{\mathrm{test}}$):}
    To evaluate generalization capability beyond historical biases, we employ a dedicated test set of approximately $1.3$ million samples collected from the future period (\textbf{Dec.~26 -- Jan.~07}). Unlike the training data, this set consists exclusively of randomized experimental data, representing an unbiased sample of the full population's covariate distribution.
\end{itemize}

\subsection{Training Settings}

\paragraph{Hyperparameter Settings. } In our real-world setting, excessive training epochs often lead to overfitting or result in the model learning to distinguish between the treatment and control groups rather than capturing the true uplift signal. Thus, we train all models for \textbf{5 epochs} to ensure robustness and use the \textbf{Adam} optimizer with \textbf{learning rate 0.001} and  \textbf{batch size 512}. We detail the other hyperparameter settings used in our experiments in Table~\ref{tab:hyperparameters}.

\begin{table}[h]
    \centering
    \caption{Hyperparameter specifications for the proposed active learning framework.}
    \label{tab:hyperparameters}
    \begin{tabular}{l c l}
        \toprule
        \textbf{Parameter} & \textbf{Value} & \textbf{Description} \\
        \midrule
        \multicolumn{3}{l}{\textit{\textbf{Acquisition Function}}} \\
        $\alpha$ (Uncertainty weight) & $0.5$ & Weight for epistemic uncertainty ($v_u$) \\
        $\beta$ (Discrepancy weight) & $1.0$ & Weight for distributional discrepancy ($d_u$) \\
        $\gamma$ (Bias weight) & $0.7$ & Weight for propensity-based bias ($o_u$) \\
        $K$ (MC Iterations) & $15$ & Number of stochastic forward passes for MC-Dropout \\
        $M$ (Sample Batch Size) & $\{2\text{k},10\text{k},20\text{k},60\text{k},100\text{k}\}$ & Number of candidates selected per active learning round \\
        \midrule
        \multicolumn{3}{l}{\textit{\textbf{Reverse Alignment Weighting}}} \\
        $w_{\text{gold}}$ & $1.0$ & Weight for samples with high counterfactual gap ($>0.5$) \\
        $w_{\text{silver}}$ & $0.2$ & Weight for samples with low counterfactual gap ($\le 0.5$) \\
        $w_{\text{obs}}$ & $0.3$ & Weight for observational samples in the training batch \\
        \midrule
        \multicolumn{3}{l}{\textit{\textbf{Domain Classifier \& Optimization}}} \\
        Representation Dim ($\mathbb{R}^{d_{rep}}$) & $128$ & Dimension of the shared representation $\phi(u)$ \\
        Hidden Units & $64$ & Hidden layer size of the Domain Classifier \\
        Dropout Rate & $0.2$ & Dropout rate for both main model and Domain Classifier \\
        Domain LR & $1e-3$ & Learning rate for the Domain Classifier \\
        Domain Steps & $100$ & Max update steps for Domain Classifier per round \\
        \midrule
        \multicolumn{3}{l}{\textit{\textbf{Active Learning Loop}}} \\
        RCT Budget ($B$) & $\{10\text{k},50\text{k},100\text{k},300\text{k},500\text{k}\}$ & Maximum total number of RCT samples allowed \\
        Max Rounds & $5$ & Number of active learning interactions \\
        Epochs per Round & $1$ & Number of training epochs within each active round \\
        \bottomrule
    \end{tabular}
\end{table}

\paragraph{Active Sampling and Training Protocol.} 
During the iterative learning phase, the model selects a subset of informative samples from $\mathcal{D}_{\mathrm{pool}}$ based on the acquisition function. We denote this actively queried subset as $\mathcal{D}_{\mathrm{rct}}$, representing the Randomized Controlled Trial data where unbiased treatment effects are observed. Consequently, the final model is trained on the union of the biased historical data and the unbiased queried data:
\begin{equation}
    \mathcal{D}_{\mathrm{train}} = \mathcal{D}_{\mathrm{obs}} \cup \mathcal{D}_{\mathrm{rct}}.
\end{equation}
This hybrid construction allows the model to leverage the scale of $\mathcal{D}_{\mathrm{obs}}$ while correcting for bias using the high-quality signals from $\mathcal{D}_{\mathrm{rct}}$.

\paragraph{Computational Infrastructure.}
All experiments are conducted on NVIDIA GPUs with mixed-precision support. We deploy computationally intensive active sampling strategies on high-performance units (L20 and RTX A6000), while assigning random sampling baselines to P40s.

\section{Evaluation Metrics}

\subsection{Normalized Area Under the Uplift Curve (AUUC)}

To evaluate the performance of our uplift model, we assess the ranking quality across the \textbf{entire test set} $\mathcal{D}_{test}$ of size $N$. Let $T_i \in \{0, 1\}$ represent the treatment assignment and $Y_i^{obs}$ denote the observed outcome for individual $i$.

The evaluation process begins by ranking the test instances based on their predicted uplift scores, $\hat{\tau}(x)$, from highest to lowest. Let the sequence of indices $i_1, i_2, \dots, i_N$ correspond to this sorted order, satisfying $\hat{\tau}(x_{i_1}) \geq \hat{\tau}(x_{i_2}) \geq \dots \geq \hat{\tau}(x_{i_N})$.

The \textbf{Uplift Curve} value at rank $k$ ($1 \leq k \leq N$), denoted as $f(k)$, quantifies the cumulative incremental gain obtained by targeting the top $k$ individuals. It is defined as \cite{pmlr-v67-gutierrez17a}:
\begin{equation}
    f(k) = \left( \frac{Y_k^T}{N_k^T} - \frac{Y_k^C}{N_k^C} \right) (N_k^T + N_k^C),
\end{equation}
where $Y_k^T = \sum_{j=1}^k Y_{i_j}^{obs} T_{i_j}$ and $N_k^T = \sum_{j=1}^k T_{i_j}$ are the cumulative outcome sum and the count of treated individuals among the top $k$ samples, respectively (with $Y_k^C$ and $N_k^C$ defined analogously for the control group using $1-T_{i_j}$).

To facilitate model comparison, we calculate the \textbf{Normalized AUUC}. This metric approximates the integral of the normalized uplift curve (over the population percentile from 0 to 1) by averaging the scaled gains across all ranks:
\begin{equation}
    \text{AUUC} \approx \int_{0}^{1} \frac{f(x)}{|f(N)|} \, dx \approx \frac{1}{N} \sum_{k=1}^N \frac{f(k)}{|f(N)|}.
\end{equation}
Here, $f(N)$ corresponds to the global lift observed on the entire test set. A higher normalized AUUC indicates that the model effectively prioritizes individuals with the highest treatment effects.

\section{Experimental Results and Analysis}
\label{sec:additional_exp}

In this appendix, we detail the experimental protocols for the three settings introduced in the main text. Furthermore, we report supplementary findings regarding data scalability and conduct a comprehensive ablation analysis.

\subsection{Exp 1: Performance under Extreme Selection Bias}

In this scenario, we focus on evaluating the model's ability to rectify learned policies under severe selection bias using a limited budget of active feedback. 
We initialize all base models using the \textbf{Biased Observational Set} ($\mathcal{D}_{\mathrm{obs}}^{\mathrm{bias}}$). 

To simulate the active learning process, we iteratively query labels from the \textbf{RCT Candidate Pool} ($\mathcal{D}_{\mathrm{pool}}$). The acquisition budget $B$ is varied across the range $\{10\mathrm{k}, 50\mathrm{k},100\mathrm{k},300\mathrm{k}, 500\mathrm{k}\}$ to analyze the performance trajectory as the ratio of experimental data increases. 
\textbf{Upon completion of retraining on the augmented dataset in each round}, the updated model is evaluated on the out-of-sample \textbf{RCT Test Set} ($\mathcal{D}_{\mathrm{rct}}^{\mathrm{test}}$). We report the performance of various base learners across these budget levels. The comparative results are visualized in Figure~\ref{fig:active_base_model} and detailed numerical metrics are provided in Table~\ref{tab:scenario1}.

\begin{table*}[t]
  \centering
  \caption{\textbf{Comparison of acquisition strategies initialized with $\mathcal{D}_{\mathrm{obs}}^{\mathrm{bias}}$.} AUUC trajectories of \emph{Active Learning} vs. \emph{Random Sampling} across increasing RCT budgets. Note that models are pre-trained on the biased observational set and evaluated on the held-out RCT test set.}
  \label{tab:scenario1}
  
  \fontsize{9pt}{11pt}\selectfont
  
  \setlength{\tabcolsep}{20pt}
  
  \begin{tabular}{lccccc}
    \toprule
    \multirow{2}{*}{\textbf{Model Setting}} & \multicolumn{5}{c}{\textbf{RCT Sample Size}} \\
    \cmidrule(lr){2-6}
     & \textbf{10k} & \textbf{50k} & \textbf{100k} & \textbf{300k} & \textbf{500k} \\
    \midrule
    
    \multicolumn{6}{l}{\textit{Backbone: DRCFR}} \\
    \hspace{1em} Random Sampling & 0.644 & 0.645 & 0.647 & 0.648 & 0.660 \\
    \hspace{1em} Active Learning (Ours) & \textbf{0.654} & \textbf{0.657} & \textbf{0.660} & \textbf{0.662} & \textbf{0.664} \\
    \addlinespace[0.5em]
    
    \multicolumn{6}{l}{\textit{Backbone: DESCN}} \\
    \hspace{1em} Random Sampling & 0.656 & 0.645 & 0.652 & 0.652 & 0.668 \\
    \hspace{1em} Active Learning (Ours) & \textbf{0.666} & \textbf{0.669} & \textbf{0.671} & \textbf{0.671} & 0.666 \\
    \addlinespace[0.5em]
    
    \multicolumn{6}{l}{\textit{Backbone: DragonNet}} \\
    \hspace{1em} Random Sampling & 0.650 & 0.664 & 0.674 & 0.648 & 0.670 \\
    \hspace{1em} Active Learning (Ours) & \textbf{0.663} & \textbf{0.670} & 0.667 & \textbf{0.669} & \textbf{0.673} \\
    
    \bottomrule
  \end{tabular}
\end{table*}

\subsection{Exp 2: Robustness to Observational Data Scale}

To validate the criticality of active sampling within our real-world production environment, we conduct a sensitivity analysis using the \textbf{Full-Scale Set} ($\mathcal{D}_{\mathrm{obs}}^{\mathrm{full}}$). 
We simulate varying data availability scenarios by down-sampling the historical business logs at different ratios $\rho$ and mixing them with the actively acquired samples:
\begin{equation}
    \mathcal{D}_{\mathrm{train}} = \text{Sample}(\mathcal{D}_{\mathrm{obs}}^{\mathrm{full}}, \rho) \cup \mathcal{D}_{\mathrm{rct}}
\end{equation}

Table~\ref{tab:scenario2} reports the AUUC performance across varying volumes and distributions of observational data. We observe that the performance advantage of our Active Sampling strategy over Random Sampling is most pronounced in \textbf{data-scarce regimes} (e.g., $0.2 \mathcal{D}_{\mathrm{obs}}^{\mathrm{full}}$) and the \textbf{highly biased} ``Group Sampling'' scenario. For instance, in the $0.2$ setting with $50$k RCT samples, active sampling surpasses the baseline by a significant margin ($0.649$ vs. $0.635$). This indicates that when observational data fails to cover the support adequately, our discrepancy-aware acquisition function effectively targets "blind spots" to repair coverage. Conversely, as the observational set approaches completeness ($0.8 \mathcal{D}_{\mathrm{obs}}^{\mathrm{full}}$), the gap naturally narrows, confirming that the proposed method yields the highest marginal utility when pre-existing knowledge is limited.

\begin{table*}[h]
  \centering
  \caption{AUUC results of the DRCFR model: Comparison between \emph{Active Sampling} and \emph{Random Sampling} under different proportions of Observational (OBS) Data.}
  \label{tab:scenario2}
  
  \fontsize{9pt}{11pt}\selectfont
  
  \setlength{\tabcolsep}{14pt}
  
  \begin{tabular}{llccccc}
    \toprule
    \multirow{2}{*}{\textbf{Data Setting}} & \multirow{2}{*}{\textbf{Sampling}} & \multicolumn{5}{c}{\textbf{RCT Sample Size}} \\
    \cmidrule(lr){3-7}
     & & \textbf{10k} & \textbf{50k} & \textbf{100k} & \textbf{300k} & \textbf{500k} \\
    \midrule
    
    \multirow{2}{*}{$0.2 \mathcal{D}_{\mathrm{obs}}^{\mathrm{full}} + \mathcal{D}_{\mathrm{rct}}$} 
      & Active & 0.639 & \textbf{0.649} & \textbf{0.654} & \textbf{0.648} & \textbf{0.656} \\
      & Random & \textbf{0.640} & 0.635 & 0.620 & 0.642 & 0.648 \\
    \midrule
    
    \multirow{2}{*}{$0.3 \mathcal{D}_{\mathrm{obs}}^{\mathrm{full}} + \mathcal{D}_{\mathrm{rct}}$} 
      & Active & \textbf{0.650} & \textbf{0.652} & \textbf{0.658} & \textbf{0.661} & \textbf{0.656} \\
      & Random & 0.633 & 0.638 & 0.635 & 0.650 & 0.653 \\
    \midrule
    
    \multirow{2}{*}{$0.5 \mathcal{D}_{\mathrm{obs}}^{\mathrm{full}} + \mathcal{D}_{\mathrm{rct}}$} 
      & Active & \textbf{0.636} & \textbf{0.648} & \textbf{0.650} & \textbf{0.649} & \textbf{0.654} \\
      & Random & 0.635 & 0.645 & 0.648 & 0.634 & 0.652 \\
    \midrule
    
    \multirow{2}{*}{$0.8 \mathcal{D}_{\mathrm{obs}}^{\mathrm{full}} + \mathcal{D}_{\mathrm{rct}}$} 
      & Active & \textbf{0.648} & 0.648 & \textbf{0.661} & \textbf{0.663} & \textbf{0.660} \\
      & Random & 0.633 & \textbf{0.653} & 0.658 & 0.656 & 0.659 \\
    \midrule
    
    \multirow{2}{*}{\shortstack[l]{Group Sampling \\ $\mathcal{D}_{\mathrm{obs}}^{\mathrm{full}} + \mathcal{D}_{\mathrm{rct}}$}} 
      & Active & \textbf{0.643} & 0.649 & \textbf{0.650} & \textbf{0.659} & \textbf{0.662} \\
      & Random & 0.640 & 0.649 & 0.646 & 0.653 & 0.654 \\
      
    \bottomrule
  \end{tabular}
  
  \begin{minipage}{0.95\linewidth}
    \footnotesize
    \textit{Note:} "Group Sampling" implies a biased observational setting where we retain 100\% of the Treated group and sample only 25\% of the Control group from $\mathcal{D}_{\mathrm{obs}}^{\mathrm{full}}$, simulating strong selection bias. Bold values indicate the higher performance in each pair.
  \end{minipage}
\end{table*}

\subsection{Exp 3: Ablation Study on Acquisition Components}

To validate the necessity of each component in our composite acquisition function $S(u)$ (defined in Eq.~\ref{eq:acq_score}), we conduct an ablation study by evaluating the performance of DRCFR under different configurations: single-component scores, pairwise combinations, and the full strategy. The results, reported in Table~\ref{tab:ablation_study}, offer several key insights:

\begin{itemize}
    \item \textbf{Uncertainty Limitations.} Relying solely on \textit{Uncertainty} ($v_u$) fails to beat the random baseline at low budgets ($10$k samples: $0.638$ vs. $0.644$). This suggests that without distributional guidance, pure uncertainty sampling risks wasting budget on outliers during the early cold-start phase.
    
    \item \textbf{Value of Overlap Deficit.} Among single components, \textit{Overlap Deficit} ($o_u$) proves most effective (e.g., $0.657$ at $100$k). This confirms that targeting samples with extreme historical bias is critical for repairing positivity violations in causal inference.

    \item \textbf{Pairwise Complementarity.} Combining any two components consistently outperforms single metrics. Notably, the combination of \textit{Discrepancy} and \textit{Overlap Deficit} ($d_u + o_u$) emerges as a strong contender (e.g., matching the full score at $500$k), indicating that simultaneously addressing domain shift and selection bias captures the majority of the information gain.
    
    \item \textbf{Synergy.} The \textbf{Full Strategy} ($v_u + d_u + o_u$) achieves the most robust performance across all sample sizes. This demonstrates that adding \textit{Uncertainty} to the mix provides the final edge, ensuring the active learner seeks not just representative and unbiased samples, but also those with high epistemic value.
\end{itemize}

\begin{table*}[h]
  \centering
  \caption{Ablation Study: Comparison of Active Sampling (different acquisition score components) vs. Random Sampling baseline. The metric reported is AUUC on the DRCFR model.}
  \label{tab:ablation_study}
  
  \fontsize{9pt}{11pt}\selectfont
  
  \setlength{\tabcolsep}{12pt}
  
  \begin{tabular}{lccccc}
    \toprule
    \multirow{2}{*}{\textbf{Acquisition Strategy}} & \multicolumn{5}{c}{\textbf{RCT Sample Size}} \\
    \cmidrule(lr){2-6}
     & \textbf{10k} & \textbf{50k} & \textbf{100k} & \textbf{300k} & \textbf{500k} \\
    \midrule
    
    \textbf{Baseline (Random)} & 0.644 & 0.645 & 0.647 & 0.648 & 0.660 \\
    \midrule
    
    \multicolumn{6}{l}{\textit{Single-Component}} \\
    \hspace{1em} Uncertainty ($v_u$) & 0.638 & 0.644 & 0.641 & 0.648 & 0.660 \\
    \hspace{1em} Discrepancy ($d_u$) & 0.641 & 0.641 & 0.649 & 0.655 & 0.652 \\
    \hspace{1em} Overlap Deficit ($o_u$)       & 0.640 & 0.647 & 0.657 & 0.658 & 0.660 \\
    \midrule
    
    \multicolumn{6}{l}{\textit{Pairwise Combinations}} \\
    \hspace{1em} $v_u + d_u$ & 0.648 & 0.653 & 0.653 & 0.651 & 0.660 \\
    \hspace{1em} $v_u + o_u$ & 0.650 & 0.654 & 0.655 & 0.659 & 0.659 \\
    \hspace{1em} $d_u + o_u$ & 0.651 & 0.656 & 0.654 & 0.652 & 0.664 \\
    \midrule
    
    \multicolumn{6}{l}{\textit{Full Strategy}} \\
    \hspace{1em} \textbf{Full Score ($v_u + d_u + o_u$)} & \textbf{0.654} & \textbf{0.657} & \textbf{0.660} & \textbf{0.662} & \textbf{0.664} \\
    
    \bottomrule
  \end{tabular}
\end{table*}

%% file: example_paper.bib
@article{cheng2021adaptive,
  title={Adaptive combination of randomized and observational data},
  author={Cheng, David and Cai, Tianxi},
  journal={arXiv preprint arXiv:2111.15012},
  year={2021}
}

@article{hatt2022combining,
  title={Combining observational and randomized data for estimating heterogeneous treatment effects},
  author={Hatt, Tobias and Berrevoets, Jeroen and Curth, Alicia and Feuerriegel, Stefan and van der Schaar, Mihaela},
  journal={arXiv preprint arXiv:2202.12891},
  year={2022}
}

@article{colnet2024causal,
  title={Causal inference methods for combining randomized trials and observational studies: a review},
  author={Colnet, B{\'e}n{\'e}dicte and Mayer, Imke and Chen, Guanhua and Dieng, Awa and Li, Ruohong and Varoquaux, Ga{\"e}l and Vert, Jean-Philippe and Josse, Julie and Yang, Shu},
  journal={Statistical science},
  volume={39},
  number={1},
  pages={165--191},
  year={2024},
  publisher={Institute of Mathematical Statistics}
}

@article{kato2024active,
  title={Active adaptive experimental design for treatment effect estimation with covariate choices},
  author={Kato, Masahiro and Oga, Akihiro and Komatsubara, Wataru and Inokuchi, Ryo},
  journal={arXiv preprint arXiv:2403.03589},
  year={2024}
}

@inproceedings{zhangactive,
  title={Active Treatment Effect Estimation via Limited Samples},
  author={Zhang, Zhiheng and Wang, Haoxiang and Li, Haoxuan and Lin, Zhouchen},
  year     = {2025},
  booktitle={Forty-second International Conference on Machine Learning}
}

@misc{wen2025enhancingtreatmenteffectestimation,
      title={Enhancing Treatment Effect Estimation via Active Learning: A Counterfactual Covering Perspective}, 
      author={Hechuan Wen and Tong Chen and Mingming Gong and Li Kheng Chai and Shazia Sadiq and Hongzhi Yin},
      year={2025},
      eprint={2505.05242},
      archivePrefix={arXiv},
      primaryClass={cs.LG},
      url={https://arxiv.org/abs/2505.05242}, 
}

@misc{gao2025causalepigpredictionorientedactivelearning,
      title={Causal-EPIG: A Prediction-Oriented Active Learning Framework for CATE Estimation}, 
      author={Erdun Gao and Jake Fawkes and Dino Sejdinovic},
      year={2025},
      eprint={2509.21866},
      archivePrefix={arXiv},
      primaryClass={stat.ML},
      url={https://arxiv.org/abs/2509.21866}, 
}

@article{addanki2022sample,
  title={Sample constrained treatment effect estimation},
  author={Addanki, Raghavendra and Arbour, David and Mai, Tung and Musco, Cameron and Rao, Anup},
  journal={Advances in Neural Information Processing Systems},
  volume={35},
  pages={5417--5430},
  year={2022}
}

@article{zhu2022active,
  title={Active learning with neural networks: Insights from nonparametric statistics},
  author={Zhu, Yinglun and Nowak, Robert},
  journal={Advances in Neural Information Processing Systems},
  volume={35},
  pages={142--155},
  year={2022}
}

@article{rubin2005causal,
  title={Causal inference using potential outcomes: Design, modeling, decisions},
  author={Rubin, Donald B},
  journal={Journal of the American statistical Association},
  volume={100},
  number={469},
  pages={322--331},
  year={2005},
  publisher={Taylor \& Francis}
}

@inproceedings{shalit2017estimating,
  title={Estimating individual treatment effect: generalization bounds and algorithms},
  author={Shalit, Uri and Johansson, Fredrik D and Sontag, David},
  booktitle={International conference on machine learning},
  pages={3076--3085},
  year={2017},
  organization={PMLR}
}

@article{Hill01012011,
    author = {Jennifer L. Hill},
    title = {Bayesian Nonparametric Modeling for Causal Inference},
    journal = {Journal of Computational and Graphical Statistics},
    volume = {20},
    number = {1},
    pages = {217--240},
    year = {2011},
    publisher = {Taylor \& Francis},
    URL = {https://doi.org/10.1198/jcgs.2010.08162},
}

@inproceedings{NEURIPS2022_675e371e,
    author = {Toth, Christian and Lorch, Lars and Knoll, Christian and Krause, Andreas and Pernkopf, Franz and Peharz, Robert and von K\"{u}gelgen, Julius},
    booktitle = {Advances in Neural Information Processing Systems},
    pages = {16261--16275},
    publisher = {Curran Associates, Inc.},
    title = {Active Bayesian Causal Inference},
    volume = {35},
    year = {2022}
}

@inproceedings{gal2016dropout,
    title={Dropout as a bayesian approximation: Representing model uncertainty in deep learning},
    author={Gal, Yarin and Ghahramani, Zoubin},
    booktitle={international conference on machine learning},
    pages={1050--1059},
    year={2016},
    organization={PMLR}
}

@article{lakshminarayanan2017simple,
    title={Simple and scalable predictive uncertainty estimation using deep ensembles},
    author={Lakshminarayanan, Balaji and Pritzel, Alexander and Blundell, Charles},
    journal={Advances in neural information processing systems},
    volume={30},
    year={2017}
}

@article{ghadiri2023finite,
  title={Finite population regression adjustment and non-asymptotic guarantees for treatment effect estimation},
  author={Ghadiri, Mehrdad and Arbour, David and Mai, Tung and Musco, Cameron and Rao, Anup B},
  journal={Advances in Neural Information Processing Systems},
  volume={36},
  pages={74180--74212},
  year={2023}
}

@inproceedings{shi2019adapting,
    title={Adapting neural networks for the estimation of treatment effects},
    author={Shi, Claudia and Blei, David M. and Veitch, Victor},
    booktitle={Advances in Neural Information Processing Systems},
    volume={32},
    year={2019}
}

@inproceedings{zhong2022descn,
    title={Descn: Deep entire space cross networks for individual treatment effect estimation},
    author={Zhong, Kailiang and Xiao, Fengtong and Ren, Yan and Liang, Yaorong and Yao, Wenqing and Yang, Xiaofeng and Cen, Ling},
    booktitle={Proceedings of the 28th ACM SIGKDD conference on knowledge discovery and data mining},
    pages={4612--4620},
    year={2022}
}

@inproceedings{cheng2022learning,
    title={Learning disentangled representations for counterfactual regression via mutual information minimization},
    author={Cheng, Mingyuan and Liao, Xinru and Liu, Quan and Ma, Bin and Xu, Jian and Zheng, Bo},
    booktitle={Proceedings of the 45th International ACM SIGIR Conference on Research and Development in Information Retrieval},
    pages={1802--1806},
    year={2022}
}

@InProceedings{pmlr-v67-gutierrez17a,
    title = 	 {Causal Inference and Uplift Modelling: A Review of the Literature},
    author = 	 {Gutierrez, Pierre and Gérardy, Jean-Yves},
    booktitle = 	 {Proceedings of The 3rd International Conference on Predictive Applications and APIs},
    pages = 	 {1--13},
    year = 	 {2017},
    editor = 	 {Hardgrove, Claire and Dorard, Louis and Thompson, Keiran and Douetteau, Florian},
    volume = 	 {67},
    series = 	 {Proceedings of Machine Learning Research},
    month = 	 {11--12 Oct},
    publisher =    {PMLR},
}
